\documentclass{llncs}

\usepackage{graphicx,color}
\usepackage{tikz}
\usepackage{amsmath,mathrsfs,amsfonts,latexsym,amssymb,verbatim}
\usepackage[mathcal]{euscript}

\newcommand{\VC}[1]{}

\newcommand{\couic}[1]{}

\newcommand{\vt}{V}
\newcommand{\dom}{\textrm{dom}}
\newcommand{\ports}{\pi}
\newcommand{\port}{{:}}

\newcommand{\tili}[1]{\widetilde{#1}}

\def\N{\mathbb{N}}

\begin{document}

\title{Reversible Causal Graph Dynamics: Invertibility, Vertex-preservation, Block representation}

\author{Pablo Arrighi\inst{1} \and Simon Martiel\inst{2} \and Simon Perdrix\inst{3}}
\institute{
Aix-Marseille Univ., LIF UMR CNRS 7279, Marseille and IXXI, Lyon, France\\
\email{pablo.arrighi@univ-amu.fr}, 
\and
INRIA Saclay, Inria Project Team DEDUCTEAM, LSV Cachan, Cachan, France\\
\email{martiel@lsv.ens-cachan.fr}
\and
CNRS, LORIA, Inria Project Team CARTE, Univ. de Lorraine, Nancy, France\\
\email{simon.perdrix@loria.fr}
}

\maketitle

\begin{abstract}
Causal Graph Dynamics extend Cellular Automata to arbitrary, bounded-degree, time-varying graphs. The whole graph evolves in discrete time steps, and this global evolution is required to have a number of physics-like symmetries: shift-invariance (it acts everywhere the same) and causality (information has a bounded speed of propagation). We add a further physics-like symmetry, namely reversibility.
\end{abstract}

\section{Introduction}

Cellular Automata (CA) consist in a $\mathbb{Z}^n$ grid of identical cells, each of which may take a state among a finite set $\Sigma$. Thus the configurations are in $\Sigma^{\mathbb{Z}^n}$. The state of each cell at time $t+1$ is given by applying a fixed local rule $f$ to the cell and its neighbours, synchronously and homogeneously across space. CA constitute the most established model of computation that accounts for euclidean space. They are widely used to model spatially distributed computation (self-replicating machines, synchronization problems\ldots), as well as a great variety of multi-agents phenomena (traffic jams, demographics\ldots).
But their origin lies in Physics, where they are commonly used to model waves or particles. And since small scale physics is understood to be reversible, it was natural to endow them with another, physics-like symmetry: reversibility. The study of Reversible CA (RCA) was further motivated by the promise of lower energy consumption in reversible computation. RCA have turned out to have a beautiful mathematical theory, which relies on topological and algebraic characterizations in order to prove that the inverse of a CA is a CA \cite{Hedlund}. The other main result is that any RCA can be expressed as a finite-depth circuits of local reversible permutations or `blocks' \cite{KariBlock,KariCircuit,Durand-LoseBlock}. If one considers that physics is reversible, this entails that RCA can indeed be implemented by physically acceptable local mechanisms --- a fact which cannot be seen directly from their local rule $\Sigma^r\to\Sigma$ description, which by definition is not injective.

Causal Graph Dynamics (CGD) \cite{ArrighiCGD,ArrighiIC,ArrighiCayleyNesme}, on the other hand, deal with a twofold extension of CA. First, the underlying grid is extended to being an arbitrary -- possibly infinite --  bounded-degree graph $G$.  Informally, this means that each vertex of the graph may take a state among a finite set $\Sigma$, so that configurations are in $\Sigma^{V(G)}$, whereas the edges of the graph represent the locality of the evolution: the next state of a vertex depends only upon the states of the vertices which are at distance at most $k$, i.e. in a disk of radius $k$,  for some fixed integer $k$. Second, the graph itself is allowed to evolve over time. 
	Informally, this means that configurations are in the union of $\Sigma^{V(G)}$ for all possible bounded-degree graph $G$, i.e. $\bigcup_G\Sigma^{V(G)}$. This twofold generalization has led to a model where the local rule $f$ is applied synchronously and homogeneously on every possible subdisk of the input graph, thereby producing small patches of the output graphs, whose union constitutes the output graph. Figure \ref{fig:CGDIdea} illustrates the concept of these CA over graphs.

\begin{figure}[h]
\begin{center}
\vspace{-0.5cm}
\includegraphics[scale=.3]{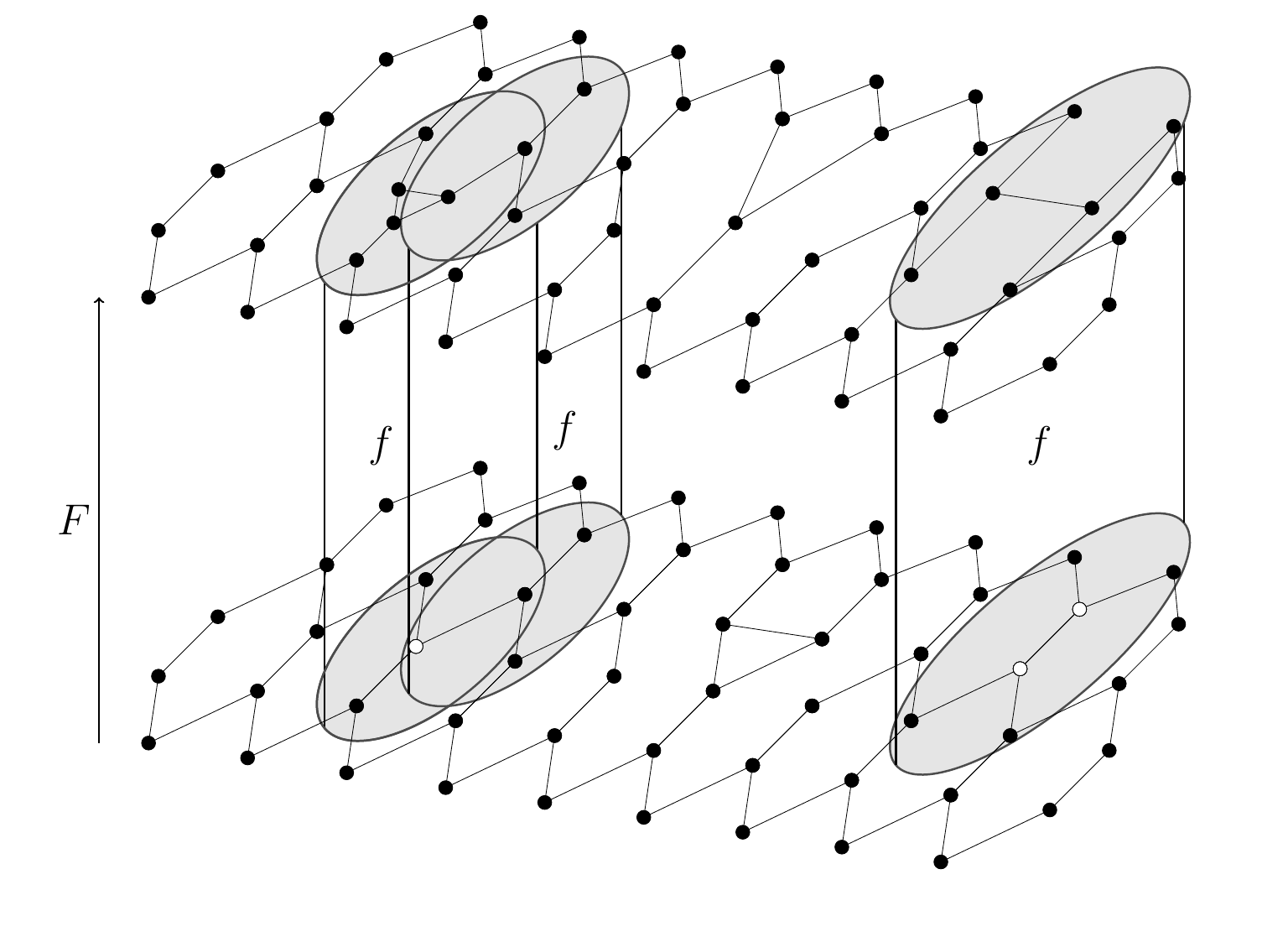}
\end{center}
\vspace{-0.7cm}
\caption{\label{fig:CGDIdea} {\em Informal illustration of Causal Graph Dynamics.}
The entire graph evolves into another according to a global function $F$. But this evolution is causal (information propagates at a bounded speed) and homogeneous (same causes lead to same effects). This has been shown to be equivalent to applying a local function $f$ to every subdisk of the input graphs, producing small output graphs whose union make up the output graph. In this paper, we take the global approach as the starting point, in order to prove that the inverse has the same properties. We then prove that there are local, reversible mechanisms for implementing $F$.
}
\end{figure}
\vspace{-0.7cm}
CGD are motivated by the countless situations in which some agents interact with their neighbours, leading to a global dynamics in which the notion of who is next to whom also varies in time (e.g. agents become physically connected, get to exchange contact details, move around\ldots). Indeed, several existing models (of physical systems, computer processes, biochemical agents, economical agents, social networks\ldots) feature such neighbour-to-neighbour interactions with time-varying neighbourhood, thereby generalizing CA for their specific sake (e.g. self-reproduction as \cite{TomitaSelfReproduction}, discrete general relativity \`a la Regge calculus \cite{Sorkin}, etc.). CGD provide a theoretical framework, for these models. Some graph rewriting models, such as Amalgamated Graph Transformations \cite{BFHAmalgamation} and Parallel Graph Transformations \cite{EhrigLowe,Taentzer,TaentzerHL}, also work out rigorous  ways to apply a local rewriting rule synchronously throughout a graph, albeit with a different, category-theory-based perspective, of which the latest and closest instance is \cite{Maignan}. But the topological approach that we follow and the reversibility question that we address have not been considered in these works. 

Indeed, this paper studies CGD in the reversible regime. Specific examples of these were described in \cite{MeyerLGA,MeyerLove}. From a theoretical Computer Science perspective, the point is therefore to generalize RCA theory to arbitrary, bounded-degree, time-varying graphs. From this perspective, our main two results consist in the generalizations of the two above-mentioned fundamental properties of RCA. The first result states that the inverse of a CGD is also a CGD. This is a non-trivial problem, for instance \cite{KariRevUndec} implies that the radius of the inverse is unbounded: there is no computable function $h$ such that for any reversible CGD of radius $r$, its inverse has a radius smaller than $h(r)$. The fact that the graph is time-varying brings up new challenges. This question was first raised in  \cite{ArrighiCGD,ArrighiIC} by Dowek and one of the authors, who proposed a first approach --- we discuss the way the present approach differs and finishes to answer this question, throughout the text. The second result shows that Reversible CGD admit a block representation, i.e. they can be implemented as a finite-depth circuit of local, reversible gates. This is a non-trivial problem: the \cite{KariCircuit} construction seems inapplicable with dynamical graphs. We manage to apply, after some work, a proof scheme which comes from Quantum CA theory \cite{ArrighiUCAUSAL}.

From a mathematical perspective, questions related to the bijectivity of CA over certain classes of graphs (more specifically, whether pre-injectivity implies surjectivity for Cayley graphs generated by certain groups \cite{Gromov}) have received quite some attention. This paper on the other hand provides a context in which to study ``bijectivity upon time-varying graphs''. 
In particular, is it the case that bijectivity  necessarily rigidifies space (i.e. forces the conservation of each vertex)? From this perspective, our main result is that Reversible CGD preserve the number of vertices of all, but a finite number of, graphs.
Again this question was first raised in  \cite{ArrighiCGD,ArrighiIC} by Dowek and one of the authors, who proposed a first approach in a more stringent setting. We discuss the greater generality of the present result within the text.

From a theoretical physics perspective, the question whether the reversibility of small scale physics (quantum mechanics, micro-mechanical), can be reconciled with the time-varying topology of large scale physics (relativity), is a topic of debate and constant investigation. This paper provides a toy, discrete, classical model where reversibility and time-varying topology coexist and interact. But ultimately, this deep question ought to be addressed in a quantum mechanical setting. Indeed, just like RCA were precursors to Quantum CA, this work paves the way for Quantum CGD. Our very recent steps in this direction are available in Pre-print \cite{ArrighiQCGD}. 

This journal paper is based upon two conference proceedings \cite{ArrighiRC,ArrighiBRCGD}. It organized as follows. After introducing our graph model, and the axiomatic definition of CGD, in Sections \ref{sec:AssociatedGraphs} and \ref{sec:Causality}, we prove in Section \ref{sec:invertibility} that invertible CGD conserve the number of vertices of almost all graphs. Section \ref{sec:reversibility} then proves that invertible CGD are reversible. Section \ref{sec:blocks} provides the block representation of Reversible CGD. We then conclude.

\section{Pointed graph modulo, paths, and operations}\label{sec:AssociatedGraphs}

\noindent {\em Why compactness?}
There are two main approaches to CA. The one with a local rule, usually denoted $f$, is the constructive one, but CA can also be defined in a more topological way as being exactly the shift-invariant continuous functions from $\Sigma^{\mathbb{Z}^n}$ to itself, with respect to a certain metric. Through a compactness argument, the two approaches are equivalent. This topological approach carries through to CA over graphs. \\
But for this purpose, one has to make the set of graphs into an appropriate compact metric space, which can only be done for certain pointed graph modulo isomorphism---referred to as generalized Cayley graphs in \cite{ArrighiCayleyNesme}.  This is worth the trouble, as the topological characterization is one of the crucial ingredients to prove that the inverse of a CGD is a CGD.\\ 
A contrario in \cite{ArrighiCGD,ArrighiIC}, Dowek and one of the authors gave a first formalism extending of cellular automata to time-varying graphs, whose set of graphs fails to be a compact metric space, and thus falls just short of achieving a proper generalisation of the classical topological definition of cellular automata.

\noindent {\em Pointed graph modulo.} Basically, the pointed graphs modulo isomorphism (or pointed graphs modulo, for short) are the usual, connected, undirected, possibly infinite, bounded-degree graphs, but with a few additional twists:
\begin{itemize}
\item[$\bullet$] The set of vertices is at most countable. Each vertex is equipped with ports. The set $\ports$ of available ports to each vertex is finite.
\item[$\bullet$] The vertices are connected through their ports, \`a la \cite{Danos200469}: an edge is an unordered pair $\{u \port a, v \port b\}$, where $u,v$ are vertices and $a,b\in \pi$ are ports. Each port is used at most once: if both $\{u\port a,v\port b\}$ and $\{u\port a, w\port c\}$ are edges, then $v=w$ and $b=c$. As a consequence the degree of the graph is bounded by $|\ports|$, 
which is crucial for compactness. We shall consider connected graphs only.
\item[$\bullet$] The graphs are rooted i.e., there is a privileged pointed vertex playing the role of an origin, so that any vertex can be referred to relative to the origin, via a sequence of ports that lead to it. 
\item[$\bullet$] The graphs are considered modulo isomorphism, so that only the relative position of the vertices can matter.
\item[$\bullet$] The vertices and edges are given labels taken in finite sets $\Sigma$ and $\Delta$ respectively, so that they may carry an internal state just like the cells of a CA. 
\item[$\bullet$] These labelling functions are partial, so that we may express our partial knowledge about part of a graph. 
\end{itemize} 
The set of all pointed graphs modulo (see Figure \ref{fig:graphs}$(c)$) having ports $\pi$, vertex labels $\Sigma$ and edge labels $\Delta$ is denoted $\mathcal{X}_{\Sigma,\Delta,\pi}$. A thorough formalization of pointed graphs modulo was given in \cite{ArrighiCayleyNesme}, and is reproduced in Appendix \ref{app:graphs} for the sake of mathematical riguour. For the sake of this paper, however, Figure \ref{fig:graphs} summarizes what there is to know about the definition of pointed graphs modulo. 

\begin{figure}
\begin{center}
\includegraphics[scale=1]{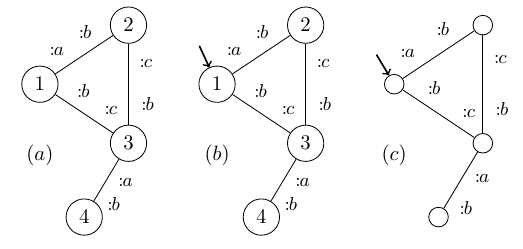}
\end{center}
\vspace{-0.7cm}
\caption{\label{fig:graphs} {\em The different types of graphs.} (a) A graph $G$. (b) A pointed graph $(G,1)$. (c) A pointed graph modulo. The latter is anonymous: vertices have no names and can only be distinguished using the graph structure.}
\end{figure}

\noindent {\em Paths and vertices.} Since we are considering pointed graphs modulo isomorphism, vertices no longer have a unique identifier, which may seem impractical when it comes to designating a vertex. Two elements come to our rescue. First, these graphs are pointed, thereby providing an origin. Second, the vertices are connected through ports, so that each vertex can tell between its different neighbours. It follows that any vertex of the graph can be designated by a sequence of ports in $(\pi^2)^*$ that lead from the origin to this vertex. For instance, say two vertices are designated by paths $u$ and  $v$, respectively. Suppose there is an edge $e=\{u\port a,v\port b\}$. Then, $v$ can be designated by the path $u.ab$, where ``$.$'' stands for the word concatenation. The origin is designated by $\varepsilon$. A thorough formalization of paths, path equivalence, and naming conventions was given in \cite{ArrighiCayleyNesme}, and is reproduced in Appendix \ref{app:paths} for the sake of mathematical riguour. Given a pointed graph modulo $X\in\mathcal{X}_{\Sigma,\Delta,\pi}$, we write $v\in X$ instead of $v\in V(X)$.\\

\noindent {\em Operations over pointed graphs modulo.}\label{sec:Operations} Given a pointed graph modulo $X$, $X^r$ denotes the subdisk of radius $r$ around the pointer. 
The pointer of $X$ can be moved along a path $u$, leading to $Y=X_u$. The pointer can be moved back where it was before, leading to $X=Y_{\overline{u}}$, where $\overline{u}$ denotes the reverse of the path $u$. We use the notation $X_u^r$ for $(X_u)^r$ i.e., first the pointer is moved along $u$, then the subdisk of radius $r$ is taken. 
A thorough formalization of these basic operations over pointed graphs modulo was given in \cite{ArrighiCayleyNesme}, and is reproduced in Appendix \ref{app:operationsmodulo} for the sake of mathematical riguour. For the sake of this paper, however, Figure \ref{fig:operations} illustrates the operations. 

\begin{figure}[h]
\begin{center}
 \includegraphics[scale=0.5]{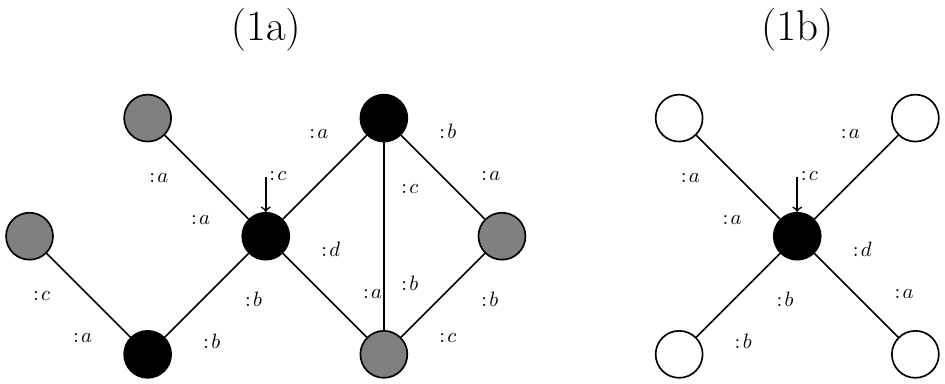}\\
 \includegraphics[scale=1]{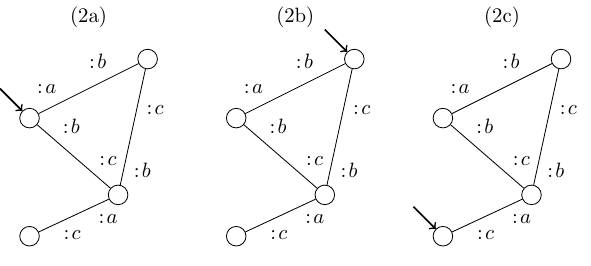}
 \end{center}
  \caption{\label{fig:operations} {\em Operations over pointed graphs modulo.} $(1)$ From $X$ to $X^0$: taking the {\em subdisk of radius} $0$. The neighbours of radius $r$ are those vertices which can be reached in $r$ steps starting from the origin. But the disk of radius $r$, written $X^r$, is the subgraph induced by the neighbours of radius $r+1$, with labellings restricted to the neighbours of radius $r$ and the edges between them. $(2a)$ A pointed graph modulo $X$. $(2b)$ $X_{ab}$ the pointed graph modulo $X$ {\em shifted} by $ab$. $(2c)$  $X_{bc.ac}$ the pointed graph modulo $X$ {\em shifted} by $bc.ac$, which also corresponds to the graph $X_{ab}$ {\em shifted} by $cb.ac$. {\em Shifting} this last graph by $\overline{cb.ac}=ca.bc$ produces the graph $(2b)$ again.}

\end{figure}

\noindent {\em Operations over graphs non-modulo.}
Sometimes we still need to manipulate usual graphs, where vertices do have names, typically in order to perform unions of two graphs in a way that specifies their overlap. In order to be able to make the union of two graphs, the graphs need to be `consistent', i.e. they must not disagree on the label of a vertex or an edge, or its connectivity through a given port. 
The set of non-modulo graphs with ports $\pi$ and vertex labels $\Sigma$ is denoted ${\cal G}_{\Sigma,\Delta,\ports}$. An example of such a graph is represented in Figure \ref{fig:graphs} $(a)$. We also use a cannonical naming function $G:{\cal X}_{\Sigma, \Delta,\ports}\rightarrow{\cal G}_{\Sigma, \Delta,\ports}$ naming each vertex of a graph modulo $X$ by the set of paths that lead to it. Notice that the graph $G(X)$ still contains the information of the position of the pointed vertex in $X$ as this is the only vertex having the empty path present in its name. We will also need an operation $u.G$ that prefixes all the vertex names in $G$, with $u$. A thorough formalization of these basic operations over graphs non-modulo was given in \cite{ArrighiCayleyNesme}, and is reproduced in Appendix \ref{app:operationsnonmodulo} for the sake of mathematical riguour. For the sake of this paper, however, Figure \ref{fig:graphop} summarizes the operations we need.

\begin{figure}
\includegraphics[scale=1]{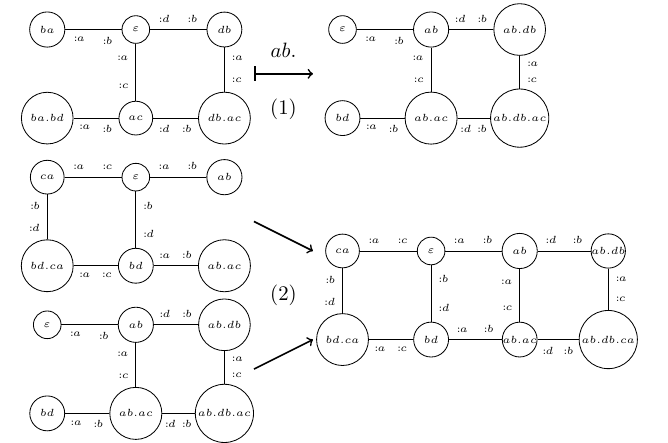}
\caption{\label{fig:graphop} Operation over graphs non-modulo. $(1)$ A prefixing of a graph by the word $ab.$. The structure of the graph is preserved, only the names of the vertices are changed. $(2)$ A graph union. Here the two graphs on the left hand side intersect on vertices $\varepsilon$, $bd$, $ab$ and $ab.ca$. As the two are consistent (e.g. in both graph, vertices $\varepsilon$ and $ab$ are connected along an $ab$ edge) their union can be computed, resulting in the right hand side graph.}
\end{figure}

\section{Causal Graph Dynamics and Invertibility}\label{sec:Causality}

\noindent {\em Causal Graph Dynamics.} We will now recall the definition of CGD. 
We reproduce the topological definition in terms of shift-invariant continuous functions, rather than a constructive definition based on a  local rule $f$ applied synchronously across space (Figure \ref{fig:CGDIdea}). The two were proved equivalent in \cite{ArrighiCayleyNesme}, but for Reversible CGD, the axiomatic approach is the right point of departure.

A crucial point in the topological  characterization of CGD is the correspondence between 
the vertices of a pointed graph modulo $X$, and those of its image $F(X)$. Indeed, on the one hand it is important to know that a given vertex $u\in X$ has become $u'\in F(X)$, e.g. in order to express shift-invariance $F(X_u)=F(X)_{u'}$, or to express continuity. But on the other hand since $u'$ is named relative to the vertex $\varepsilon$ of $F(X)$, its determination requires some knowledge of $F(X)$. Hence the need for establishing a relation between the vertices of $X$ and those of $F(X)$.

The following analogy provides a useful way of tackling this issue. Say that we were able to place a white stone on the vertex $u\in X$ that we wish to follow across evolution $F$. Later, by observing that the white stone is found at $u'\in F(X)$, we would be able to conclude that $u$ has become $u'$. This way of grasping the correspondence between an image vertex and its preimage vertex is a local, operational notion of an observer moving across the dynamics. 

\begin{definition}[Dynamics]\label{def:dynamicsmodulo}
A dynamics $(F,R_{\bullet})$ is given by
\begin{itemize}
\item[$\bullet$] a function $F\colon{\cal X}_{\Sigma,\Delta,\ports}\to{\cal X}_{\Sigma,\Delta,\ports}$;
\item[$\bullet$] a map $R_{\bullet}$, with $R_{\bullet}\colon X\mapsto R_X$ and $R_X: V(X) \to V(F(X))$.
\end{itemize}
Notice that the function $R_X$ can be pointwise extended to sets of vertices i.e., $R_X\colon{\cal P}(V(X)) \to {\cal P}(V(F(X)))$ maps $S$ to $R_X(S)=\{R_X(u)\;|\;u\in S\}$.
\end{definition}
The intuition is that $R_X$ indicates which vertices $\{u',v',\ldots\}=R_X(\{u,v,\ldots\})\subseteq V(F(X))$ will end up being marked as a consequence of $\{u,v,...\}\subseteq V(X)$ being marked. Now, clearly, the set $\{(X,S)\;|\;X\in {\cal X}_{\Sigma,\Delta,\ports}, S\subseteq V(X)\}$ is isomorphic to ${\cal X}_{\Sigma',\Delta,\ports}$ with $\Sigma'=\Sigma\times\{0,1\}$. Hence, we can define the function $F'$ that maps $(X,S)\cong X'\in {\cal X}_{\Sigma',\Delta,\ports}$ to $(F(X),R_X(S))\cong F'(X')\in {\cal X}_{\Sigma',\Delta,\ports}$, and think of a dynamics as just this function $F':{\cal X}_{\Sigma',\Delta,\ports}\to {\cal X}_{\Sigma',\Delta,\ports} $. 

\noindent Next, continuity is the topological way of expressing causality, i.e. bounded speed of propagation of information:
\begin{definition}[Continuity]\label{def:continuitymodulo}
A dynamics $(F,R_{\bullet})$ is said to be {\em continuous} if for any $X$ and any $m\ge 0$, there exists $n\ge 0$ such that for every $Y$, $X^n=Y^n$ implies both
\begin{itemize}
\item[$\bullet$] $F(X)^m=F(Y)^m$.
\item[$\bullet$] $\dom\,R_{X}^m\subseteq V(X^n)$, $\dom\,R_{Y}^m\subseteq  V(Y^n)$, and $R_{X}^m=R_{Y}^m$.
\end{itemize}
where $R_{X}^m$ denotes the partial map obtained as the restriction of $R_X$ to the codomain $F(X)^m$, using the natural inclusion of $F(X)^m$ into $F(X)$.
\end{definition}
In the $F':{\cal X}_{\Sigma',\Delta,\ports}\to{\cal X}_{\Sigma',\Delta,\ports}$ formalism, the two above conditions are equivalent to just one: $F'$ continuous.

\begin{lemma}[White stone]
A dynamics $(F,R_{\bullet})$ is continuous if and only if its corresponding $F'$ is continuous. 
\end{lemma}
\begin{proof}
\begin{align*}
&\forall X',\,\forall m,\,\exists n\,/\,\forall Y',\,  [{X'}^n={Y'}^n \Rightarrow F'(X')^m=F'(Y')^m] \\
\Leftrightarrow\quad & \forall D,\,\forall U,\,\forall X,\,\forall S\,\,\forall m,\,\exists n\,/\,\forall Y',\,\\
&[({X}^n={Y}^n=D\,\wedge\, S^n=T^n=U) \Rightarrow\\
&\quad(F(X)^m=F(Y)^m=F(D)^m)\,\wedge\,R_X^m(S)=R_Y^m(T)=R_D^m(U))]\\ 
\Leftrightarrow\quad & \forall D,\,\forall U,\,\forall X,\,\forall S\,\,\forall m,\,\exists n\,/\,\forall Y,\,\\
&[({X}^n={Y}^n=D\,\wedge\, S^n=T^n=U) \Rightarrow\\
&\quad(F(X)^m=F(Y)^m)\,\wedge\,R_X^m(S)=R_Y^m(T)=R_D^m(U))\\
&\quad\wedge\,\dom\,R_{X}^m=V(U)=V(X^n) \,\wedge\,\dom\,R_{Y}^m= V(U)=V(Y^n)]\\
\Leftrightarrow\quad & \forall X,\,\forall m,\,\exists n\,/\,\forall Y,\,\\
&[{X}^n={Y}^n \Rightarrow (F(X)^m=F(Y)^m)\\
&\quad\wedge\,\dom R_{X}^m=V(X^n)\,\wedge\,\dom R_{Y}^m=V(Y^n)\,\wedge\,R_X^m=R_Y^m]
\end{align*}
\end{proof}

\noindent Recall that the reason why continuity is the topological way of expressing causality, is because it turns out to be equivalent to uniform continuity:
\begin{definition}[Uniform continuity]\label{def:continuitymodulo}
A dynamics $(F,R_{\bullet})$ is said to be {\em uniformly continuous} if for any $m\ge 0$, there exists $n\ge 0$ such that for every $X, Y$, $X^n=Y^n$ implies both
\begin{itemize}
\item[$\bullet$] $F(X)^m=F(Y)^m$.
\item[$\bullet$] $\dom\,R_{X}^m\subseteq V(X^n)$, $\dom\,R_{Y}^m\subseteq  V(Y^n)$, and $R_{X}^m=R_{Y}^m$.
\end{itemize}
where $R_{X}^m$ denotes the partial map obtained as the restriction of $R_X$ to the codomain $F(X)^m$, using the natural inclusion of $F(X)^m$ into $F(X)$.
\end{definition}
In the $F':{\cal X}_{\Sigma',\Delta,\ports}\to{\cal X}_{\Sigma',\Delta,\ports}$ formalism, the two above conditions are equivalent to just one: $F'$ uniformly continuous.

\noindent Recall also that the reason for this equivalence between continuity and uniform continuity is the compactness of ${\cal X}_{\Sigma',\Delta,\ports}$ as a metric space, which allows for a direct application of Heine's theorem, as is extensively discussed in \cite{ArrighiCayleyNesme}. A contrario the formalism of \cite{ArrighiCGD,ArrighiIC} by Dowek and one of the authors lacked a compact metric space, and such results had to be reproven by hands. 

\noindent We now express the fact that the same causes lead to the same effects:
\begin{definition}[Shift-invariance]
A dynamics $(F,R_{\bullet})$ is said to be {\em shift-inva\-riant} if for every $X$,  $u\in X$, and $v\in X_u$, 
\begin{itemize}
\item[$\bullet$] $F(X_u)=F(X)_{R_X(u)}$
\item[$\bullet$] $R_X(u.v)=R_X(u).R_{X_u}(v)$.
\end{itemize}
\end{definition}
The second condition expresses the shift-invariance of $R_{\bullet}$ itself. Notice that  $R_X(\varepsilon)=R_X(\varepsilon).R_X(\varepsilon)$; hence $R_X(\varepsilon)=\varepsilon$.

\noindent Finally we demand that the graphs does not expand in an unbounded manner:
\begin{definition}[Boundedness]\label{def:boundednessmodulo}
A dynamics $(F,R_{\bullet})$ 
is said to be {\em bounded} if there exists a bound $b$ such that for any $X$ and any $w'\in F(X)$, there exist $u'\in \textrm{Im}(R_X)$ and $v'\in F(X)_{u'}^b$ 
such that $w'=u'.v'$.
\end{definition}

\noindent Putting these conditions together yields the topological definition of CGD:
\begin{definition}[Causal graph dynamics]\label{def:causal}
A CGD is a shift-invariant, continuous, bounded dynamics.
\end{definition}

\noindent\emph{Example: inflating grid.} An example of CGD is given in Figure \ref{fig:inflatingglobal}. In the inflating grid example, each vertex gives birth to four distinct vertices, such that the structure of the initial graph is preserved, but inflated. The graph has maximal degree $4$, and set of ports $\pi=\{a,b,c,d\}$, its vertices 
are labelled black or white.

\begin{figure}[h]
\centerline{\includegraphics[scale=0.98]{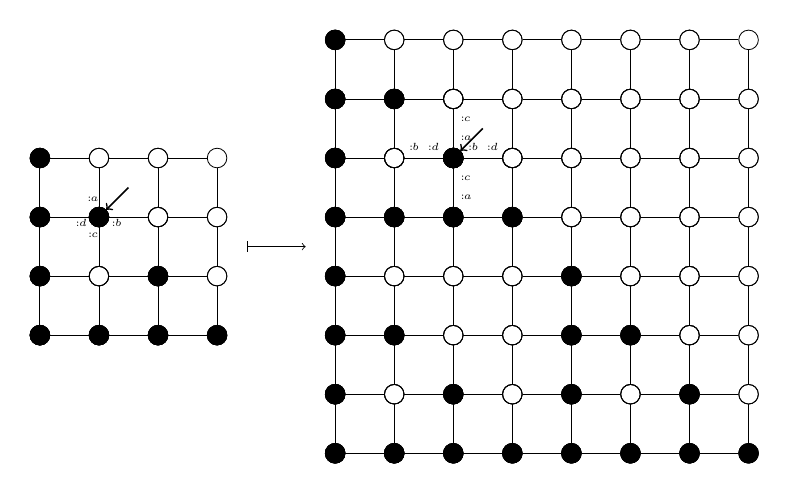}}
\caption{\label{fig:inflatingglobal} {\em The inflating grid example.} Each vertex splits into 4 vertices. The structure of the grid is preserved. For this precise graph, all edges are connected to ports as stipulated on the pointed vertex (port $\port a$ on top, $\port b$ on the right, $\port c$ on the bottom and $\port d$ on the left).}
\end{figure}

\medskip
\noindent {\em Invertibility.} Invertibility is imposed in the most general and natural fashion.

\begin{definition}[Invertible dynamics]
A dynamics $(F,R_\bullet)$ is said to be invertible if $F$ is a bijection over $ {\cal X}_{\Sigma,\Delta,\ports}$.
\end{definition}
\noindent\emph{Example: moving head.} Figure \ref{fig:poil} is an example of invertible CGD. In this example, a vertex, representing the head of an automaton, is moving along a path graph, representing a tape. The path graph is built using $ab-$edges, while the head is attached using either a $cc-$edge if it is travelling forward along the $ab-$edges, or $dd-$edges if it is travelling backwards. The transformation can be completed into a bijection over the entire set of graphs with $\pi=\{a,b,c,d\}$. It then accounts for several heads, etc. The resulting transformation is continuous, as the moving heads travel at speed one along the tape, and shift-invariant as it is possible to build a $R_\bullet$ operator verifying the right commutation properties.

\begin{figure}
\vspace{-0.5cm}
\begin{tabular}{lll}
\includegraphics[scale=0.32]{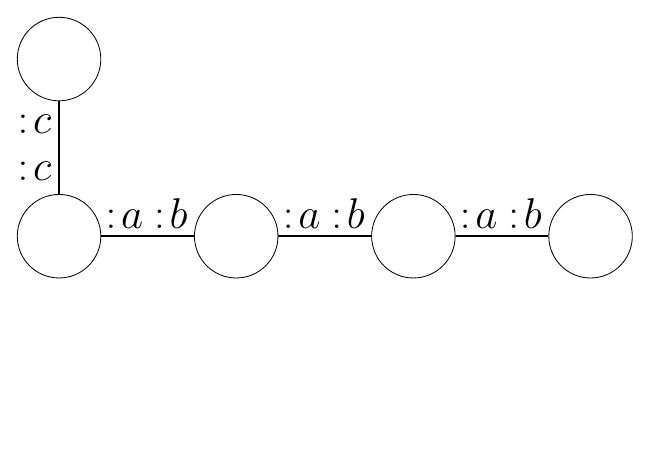} &
~~~~~\includegraphics[scale=0.32]{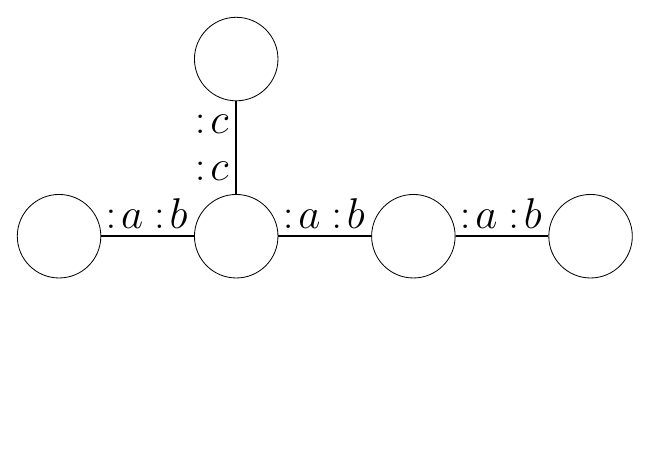}~~~~~ &
\includegraphics[scale=0.32]{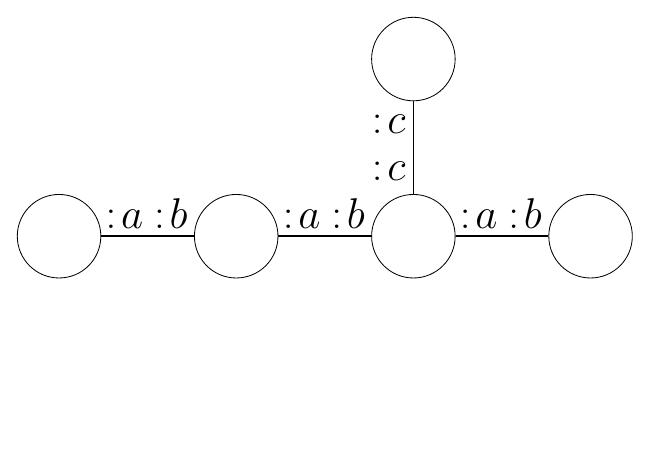}\vspace{-1cm} \\
~~(1)&~~~~~~~(2)&~~(3)\\

\includegraphics[scale=0.32]{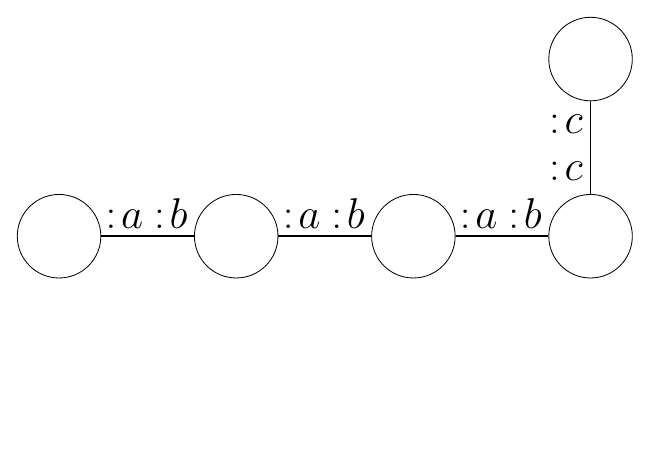} &
~~~~~\includegraphics[scale=0.32]{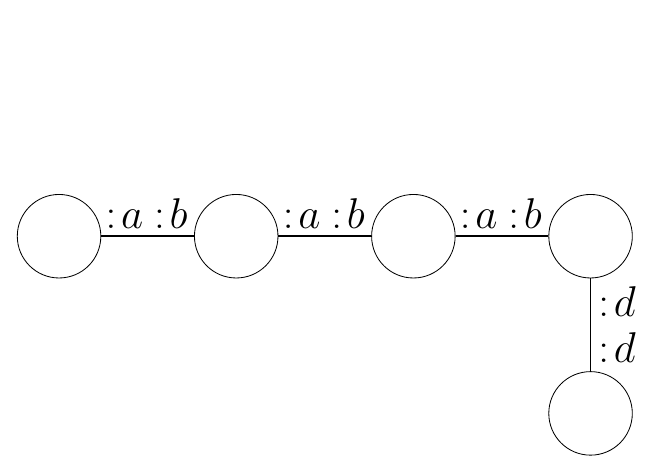} ~~~~~&
\includegraphics[scale=0.32]{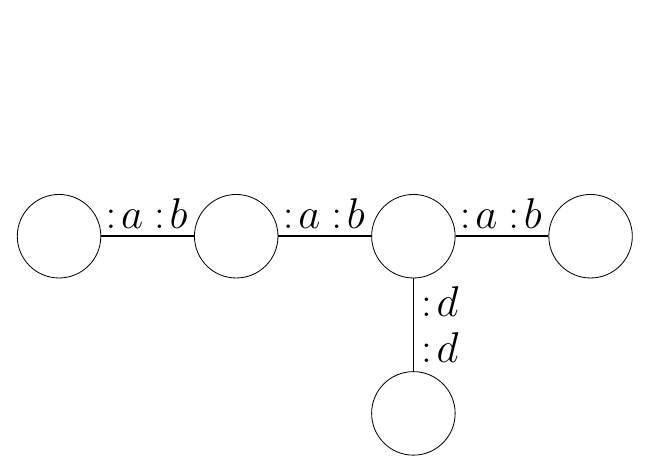} \vspace{-1cm} \\
~~(4)&~~~~~~~(5)&~~(6)\\
\\
\end{tabular}
%
\caption{\label{fig:poil} {\em The moving head example.} In this example, a head is attached, via a $cc-$ edge, to a ``tape'' formed by $ab-$edges. It moves forward until it reaches the end of the line. It then changes the attaching ports to $dd$ and moves backwards. $(1)$ to $(6)$ represent consecutive configurations.}
\end{figure}
\noindent\emph{Example: Turtle dynamics.} The turtle dynamics simply oscillates between the two pointed graphs modulo of degree $1$. Figure \ref{fig:rofturtle} describes its associated $R_\bullet$ operator.

\begin{figure}[h]
\begin{center}
\includegraphics[scale=1]{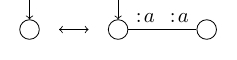}
\end{center}
\vspace{-0.7cm}
\caption{\label{fig:turtle} {\em The turtle example} has the two above pointed graphs modulo to oscillate between one another. The two vertices of the RHS are shift-equivalent, i.e. pointing the graph upon one or the other does not change the graph.}
\end{figure}

\begin{figure}[h]
\begin{center}
\includegraphics[scale=1]{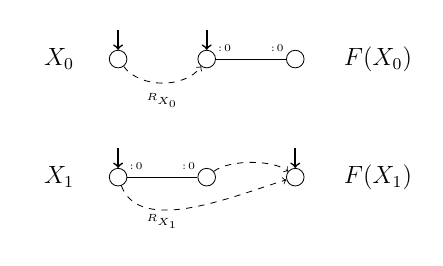}
\end{center}
\vspace{-0.7cm}
\caption{\label{fig:rofturtle} The $R_\bullet$ operator associated to the turtle dynamics.}
\end{figure}

\medskip
\noindent {\em Restricting CGD\ldots or extending CA?} Notice that, because CGD are strictly more powerful than CA, it can be a tricky task to associate to a given CA a single, canonical CGD. Consider, for instance, a one-dimensional CA with internal state in a finite set $\Sigma$. The natural set of graphs to consider would be the complete set ${\cal X}_{\Sigma,\Delta,\ports}$, with $\pi=\{a,b\}$. But in this natural set not all graphs can be interpreted as CA configurations. In CA configurations each vertex has degree exactly $2$, and each edge is of the form $ab$. There are two distinct ways of tackling this issue:
\begin{itemize}
\item The canonical approach is to work with a restricted CGD, that is no longer defined on the complete set of graphs ${\cal X}_{\Sigma,\Delta,\ports}$, but on a subset of it. If this subset is characterized by forbidding the occurrence of certain disks, then it will still be a compact metric space, and theorems of this paper will carry through, except perhaps that of the block representation. In the specific case of CA, however, forbidding the occurrence of disk that have non degree $2$ vertices and non $ab$ edges amounts to just restricting to CA configurations. Thus, these restricted CGD are just CA, for which we know there are equivalent theorems \cite{Hedlund,KariCircuit,ArrighiBLOCKREP}. This approach is canonical since it requires no additional choice.
\item The non-canonical approach is to extend CA to become a fully-defined CGD over ${\cal X}_{\Sigma,\Delta,\ports}$. It is not clear to us at this stage whether this can always be done whilst perserving invertibility, and keeping labels within $\Sigma$. If $\Sigma$ is allowed to be extended, then we conjecture that this is possible, via a prior step of \cite{ArrighiBLOCKREP}. In any case there are many possible choices for such an extension, this approach is non-canonical.  
\end{itemize}

\section{Invertibility and almost-vertex-preservation}	\label{sec:invertibility}

Recall that, in general, CGD are allowed to transform the graph, not only by changing internal states and edges, but also by creating or deleting vertices. Since invertibility imposes information-conservation, one may wonder whether invertible CGD are still allowed to create or delete vertices. They are, as shown by Figure \ref{fig:turtle}. One notices, however, that the RHS of this example features shift-equivalent vertices:
\begin{definition}[Shift-equivalent vertices]\label{def:shifteqvert}
Let $X\in {\cal X}_{\Sigma,\Delta,\ports}$ and let $u,v\in X$. We say that $u$ and $v$ are \emph{shift-equivalent}, denoted $u\sim_X v$, if $X_u=X_v$.
A graph is called \emph{asymmetric} if it has only trivial (i.e. size one) shift-equivalence classes.
\end{definition}

One can show that all the shift-equivalence classes of a pointed graph modulo  have the same size. Intuitively, given  two shift-equivalent vertices $u,v$ and a third vertex $w$, since there is a path from $u$ to $w$,  moving from $v$ along the same path leads to a vertex equivalent to $w$.

\begin{lemma}[Shift-equivalence classes isometry]\label{lemma:iso}
Let $X\in{\cal X}_{\Sigma,\Delta,\ports}$ be a graph. If $C_1\subseteq{V(X)}$ and $C_2\subseteq{V(X)}$ are two shift-equivalence classes of $X$, then $C_1$ and $C_2$ have the same cardinality.
\end{lemma}

\begin{proof}
Consider two equivalent and distinct vertices $u$ and $v$ in $X$. Consider a path $w$. The vertices $u.w$ and $v.w$ are distinct and equivalent. More generally, if we have $n$ equivalent distinct vertices $v_1,...,v_n$, any vertex $u=v_1.w$ will be equivalent to $v_2.w,...,v_n.w$ and distinct from all of them, hence all the equivalence classes have the same cardinality.\hfill $\Box$
\end{proof}

\begin{lemma}[Structure of symmetric graphs]\label{lemma:strucsym}
Let $X\in{\cal X}_{\Sigma,\Delta,\ports}$ be a symmetric graph. Then
$$X=\bigcup_{u\in T}u.G$$
with 
\begin{itemize}
\item $T$ vertex-transitive. 
\item $V\subseteq V(X)$ such that $\varepsilon\in V$ and $\{u.V\}_{u\in T}$ is a partition of $V(X)$.
\item $G=G(X_V^0)$.
\item $w\sim w'$ if and only if $w=u.v$, $w'=u'.v$, $u,u'\in T$, and $v\in G$.
\end{itemize}
\end{lemma}
\begin{proof}
From symmetry there are $u\neq v$ such that $X_u=X_v$. Write
$u\stackrel{\neq}{\sim}v$. This entails $u.w\stackrel{\neq}{\sim}v.w$, and taking $w=\overline{u}$ we have $\varepsilon\stackrel{\neq}{\sim}v.\overline{u}$. Call $C$ the equivalence class of $\varepsilon$.\\
$[$Constructing $T$ $]$.  
Let $H$ be a pointed graph non-modulo such that $V(H)=C$. We construct its set of edges as follows. First, for every $u,v\in C$ we define the set of straight paths between them to be $$S(u,v)=\{w\in X,|\,u.w=v\wedge\nexists x,y / x.y=w\wedge u.x\in C\}.$$
A path is less than another if it is shorter in length, or equal in length and smaller lexicographically. $s(u,v)$ is the minimal $S(u,v)$. The left (resp. right) of a path is the word made of every odd (resp. even) letter.
$$\{u:\textrm{left}(s(u,v)),v:\textrm{right}(s(u,v))\}\in E(H)\,\Leftrightarrow\, S(u,v)\neq \emptyset.$$
The pointer is at $\varepsilon$. That defines $H$ and $T=\widetilde{(H,\varepsilon)}$, from $X$ and $C$ the equivalence class of $\varepsilon$.\\

$[$Vertex transitivity of T$]$ Consider $u\in C$. We can do the above procedure starting from $u$ and obtain a graph modulo $T(u)$. Since $u$ and $\varepsilon$ are symmetric and since the construction of $H$ does not privilege any vertex, we have that $T(\varepsilon) = T(u)$. By construction we have that $T(\varepsilon)=\widetilde{(H,\varepsilon)}$ and $T(u)=\widetilde{(H,u)}=\widetilde{(H,\varepsilon)}_u$. Hence, we have that $T(\varepsilon)=T(\varepsilon)_u$.


$[$Constructing $G$ $]$. Let $v\notin C$,
$$V=\big\{v\,|\,\forall u\in C,\,s(\varepsilon,v)\textrm{ exists and is smaller than }s(u,v)\}\cup\{\varepsilon\}.$$
Notice that $X_V^0$ is connected. Indeed, if $w$ is a vertex along $s(\varepsilon,v)$, then $w\in V$. Let $G=G(X_V^0)$. That defines $G$ from $X$ and $C$ the equivalence class of $\varepsilon$.\\
Starting from some $u$ in $T$ instead, we would have obtained some $V(u)$ and some $G(u)$. But notice that $u\in V(T)\Rightarrow u\in C\Rightarrow X_u=X$ implies $V(u)=u.V$. Indeed, $v\in V(\varepsilon)$ is equivalent to $u.v\in V(u)$, because $s(\varepsilon,v)$ minimal is equivalent to $s(u,u.v)$ minimal. Moreover, since for all $v\in X$, $X_{u.v}^r=X_v^r$, we have $G(u)=G(X_{V(u)}^0)=G(X_{u.V}^0)=G(X_V^0)=G$.\\
Moreover, there is always some $u\in T$ that has the smallest straight path to $v$, therefore every $v$ must belong to some $V(u)$.\\
Finally, every edge of $X$ has one of its vertices in some $V(u)$, and must therefore belong to some $G(u)$.\\
Therefore, altogether
$$X=\bigcup_{u\in T}u.G.$$
\\
$[$ $w\sim w'$ if and only if $w=u.v$, $w'=u'.v$, $u,u'\in T$, and $v\in G$$]$ Consider $w\sim w'$. Consider $u\in T$ such that $w\in G(u)$, $u' \in T$ such that $w'\in G(u')$. Consider $v$ and $b'$ the shortest paths equivalent to $w$ and $w$. By definition, $v$ and $v'$ belong to $G$, and, since $w \sim w'$ we have that  $v = v'$. By minimality of $v$ we have that $w = u.v$ and $w'=u'.v$. Conversely if $u$ and $u'$ are in $T$ and $v$ is in $G$ then $u.v \sim u'.v$ by construction of $T$.

\qed
\end{proof}

\begin{definition}[Asymmetric extension]\label{def:asymext}
Given a finite symmetric graph $X\in{\cal X}_{\Sigma,\Delta,\ports}$, let $V$ be defined as in Lemma \ref{lemma:strucsym}. We obtain an {\em asymmetric extension} ${}^\Box X$ by either:
\begin{itemize}
\item Choosing a vertex $w\in V$ having a free port and connecting an extra vertex $w.e$ onto it. 
\item Or choosing vertex $w\in V$ that is part of a cycle, removing an edge $e$ of the cycle $w$ that was connecting $w$ and $w'$, and adding the two extra vertices $w.e$ and $w'.\overline{e}$, having the same label as the removed edge.
\end{itemize}
\end{definition}

\begin{lemma}[Asymetry of asymetric extension]\label{lem:primalext}
Given a finite symmetric graph $X\in{\cal X}_{\Sigma,\Delta,\ports}$, its asymmetric extension ${}^\Box X$ is asymmetric, and $|{}^\Box X| \leq |X| + 2$.
\end{lemma}
\begin{proof}
We will refer to $T$, $V$ and $G$ of Lemma \ref{lemma:strucsym}. We call $w$ the vertex of $V$ upon which edge $e$ is added to take $X$ into ${}^\Box X$.\\
First we prove that old symmetries are broken by the extension. Symmetries in $X$ were of the form $u.v\sim_X u'.v$ with $u,u'\in T$, $u\neq u'$ and $v\in G$. Consider $x=\overline{v}.\overline{u}.w$, so that $u.v.x=w$ and $u'.v.x=u'.\overline{u}.w$. If $w$ had a free port we have $x.e\in {}^\Box X_{u.v}$ but $x.e\notin {}^\Box X_{u'.v}$. If $w$ had no free port we have that $x.e$ has no further edge $e'\neq e$ in ${}^\Box X_{u.v}$, but has a further edge $e'\neq e$ in ${}^\Box X_{u'.\overline{u}}$. In both cases $u.v\not\sim_{{}^\Box X} u'.v$.

Second we prove that no new symmetry has been created by the extension. In $X$, $u.v\not\sim_X u'.v'$ implied the existence of two witness paths $x,y\in X$ such that at least one of the following holds
\begin{itemize}
\item[$(i)$] The label of the vertex reached by path $x$ is different in $X_{u.v}$ from what it is in $X_{u'.v'}$.
\item[$(ii)$] The label of the last edge borrowed by path $x$ is different in $X_{u.v}$ from what it is in $X_{u'.v'}$.
\item[$(iii)$] The path $x$ has an edge $f$ in $X_{u.v}$ but not $X_{u'.v'}$, or the converse.
\item[$(iv)$] The paths $x$ and $y$ do not reach the same vertex in $X_{u.v}$ but they do in $X_{u'.v'}$, or the converse.
\end{itemize}
If $w$ is in a cycle, the witness paths can be chosen so that $w$ is not a strict prefix of $x,\, y$, which we do in order to garantee that $x,\, y\in {}^\Box X$.
The existence of $x$ as in $(i)$ and $(ii)$ is clearly unaffected by the asymetric extension, because it does not act on labels. Therefore in these cases, $u.v\not\sim_{{}^\Box X} u'.v'$ still holds. We can exclude those cases.\\
The existence of $x,y$ as in $(iii)$ or $(iv)$ is also clearly unaffectected when both $u.v.x\notin V$ and $u'.v'.x\notin V$, because the asymetric extension only acts on $V$. The case when both $u.v.x\in V$ and $u'.v'.x\in V$ is not problematic either, because the symmetries of $X$ tell us that there are $x',y'$ meeting the same condition and such that $u.v.x'\notin V$ and $u'.v'.x'\notin V$, namely $x'=\overline{x}.\overline{v}.u''.v.x,\,y'=\overline{y}.\overline{v}.u''.v.y$. We can exclude those cases.\\
From now on suppose, without loss of generality, that $u.v.x\in V$ and $u'.v'.x\notin V$.\\
Then, the existence of $x,y$ as in $(iii)$ or $(iv)$ is clearly unaffected in the direct (aka non-converse) subcase. But the converse subcase is not problematic either, because the symmetries of $X$ tell us that there are $x', y'$ meeting the same condition and such that $u.v.x'\notin V$ and $u'.v'.x'\in V$, namely $x'=\overline{x}.\overline{v}.u'.v'.x,y'=\overline{y}.\overline{v}.u'.v'.y$. We covered all cases.\\
\end{proof}

Moreover, we can show that creation or deletion of vertices by invertible CGD must respect the shift-symmetries of the graph. 

\begin{lemma}[Invertible CGD preserves shift-equivalence classes] \label{lem:invert1}
Let $(F,R_\bullet)$ be a shift-invariant dynamics over $\mathcal{X}_{\Sigma,\Delta,\pi}$, such that $F$ is a bijection. Then for any $X$ and any $u,v\in X$,  $u \sim_X v $ if and only if $R_X(u) \sim_{F(X)} R_X(v)$.
\end{lemma}
\begin{proof}
$u\sim_X v$ expresses $X_u=X_v$, which by bijectivity of $F$ is equivalent to  $F(X_u)=F(X_v)$ and hence $F(X)_{R_X(u)}=F(X)_{R_X(v)}$. This in turn is expressed by $R_X(u) \sim_{F(X)} R_X(v)$. \hfill $\Box$
\end{proof}

\begin{lemma}[Invertible dynamics preserve the number of shift-equivalence classes]\label{lem:invert2}
Let $(F,R_\bullet)$ be a causal graph dynamics over $\mathcal{X}_\pi$, such that $F$ is a bijection. Then for all finite graph $X$, we have $|X_{\diagup_\sim}|=|F(X)_{\diagup_\sim}|$.
\end{lemma}
\begin{proof}
\noindent $\bullet$  $|X_{\diagup_\sim}|\leq|F(X)_{\diagup_\sim}|$: 
 Let us assume that there exists a graph $X$ such that $|X_{\diagup_\sim}|>|F(X)_{\diagup_\sim}|$. As $F$ is shift-invariant, to each possible way of pointing the graph $X$ correspond a way of pointing the graph $F(X)$. As there are less ways of pointing $F(X)$ than ways of pointing $X$, there must exist $u$ and $v$ in $X$ such that $F(X_u)=F(X_v)$ and $u$ and $v$ non-equivalent, which contradicts the injectivity of $F$.
 
\noindent $\bullet$  $|X_{\diagup_\sim}|\geq|F(X)_{\diagup_\sim}|$:  Let us assume that there exists a graph $Y$ such that $|Y_{\diagup_\sim}|<|F(Y)_{\diagup_\sim}|$. Let us consider the graph sequence $(Y(k))_{k\in\mathbb{N}}$, $Y(k)=F^{-k}(Y)$. The natural number sequence $|Y(k)_{\diagup_\sim}|$ is decreasing (else it would contradict the first $\bullet$), thus it converges and reaches its limit. More precisly, there exists a rank $n_0$ such that for all $n>n_0$, $|Y(n+1)_{\diagup_\sim}|=|Y(n)_{\diagup_\sim}|=|F(Y(n+1))_{\diagup_\sim}|$. This new sequence $(Y(n))$ is infinite and, by bijection of $F$, $|Y(n)|$ is unbounded. Let us extract an infinite subsequence $(Y(l))$ such that, for all $l$, $|Y(l)| > |F(Y(l))|$ and $|Y(l+1)| > |Y(l)|$. All the graphs in this sequence have the same number $b$ of shift-equivalence classes. \\
Let $r$ be the radius of a local rule inducing $F$. We can now extract a subsequence $Y(m)$ such that all graphs of this sequence contain exactly the same neighborhoods of radius $r$. This is possible as there are finitely many neighbourhoods of radius $r$ and only $b$ neighbourhoods to pick. As they contain the same neighbourhoods, $f$ will act the same on all these graphs and thus they ought to decrease their sizes at the same speed: $|F(Y(m))| < |Y(m)|$ and $|F(Y(m))| = \alpha.|Y(m)|$ for some $\alpha < 1$ depending only on $F$.

We can now apply an asymetric extension on a large enough graph $Y(m)$, yielding an asymetric graph ${}^\Box Y(m)$. Using the bounded inflation lemma, we can bound the size of this graph: $|F({}^\Box Y(m))| < |F(Y(m))| + C$ for some constant $C$ depending only on $F$. Thus we have $|F({}^\Box Y(m))| < \alpha.|Y(m)| + C $. Since $\alpha < 1$, we can derive the inequality $|F({}^\Box Y(m))| < |Y(m)| = |{}^\Box Y(m)|$ for a large enough $Y(m)$. This inequality contradicts the first $\bullet$, thus $|X_{\diagup_\sim}|\geq|F(X)_{\diagup_\sim}|$. 
\end{proof}

Shift-symmetry is fragile however, and can be destroyed by adding a few vertices to a graph:


%

Using this fact, one can show that the cases of node creation and deletion in invertible CGD are all of finitary nature, i.e. they can no longer happen for large enough graphs. Indeed, by supposing a big enough graph $X$ whose order is changed through the application of an invertible CGD, and then looking at what would happen to its asymetric extension ${}^\Box X$, we can show that this would contradict continuity. We obtain:

\begin{theorem}[Invertible implies almost-vertex-preserving]\label{th:preserv}
Let $(F,R_\bullet)$ be a CGD over $\mathcal{X}_{\Sigma,\Delta,\ports}$, such that $F$ is a bijection. Then there exists a bound $p$, such that for any graph $X$, if $|V(X)| > p$ then $R_X$ is bijective.
\end{theorem}
\begin{proof} When $|\pi|\le 1$, $\mathcal{X}_{\Sigma,\Delta,\ports}$ is finite so the theorem is trivial. So we assume in the rest of the proof that $|\pi|>1$. \\
\noindent[Finite graphs]  First we prove the result  for any finite graph. 
By contradiction,  assume that there exists a sequence of finite graphs $(X(n))_{n\in\mathbb{N}}$ such that $|V(X(n))|$ diverges and such that for all $n$, $R_{X(n)}$ is not bijective. As this sequence is infinite, we have that one of the two following cases is verified for an infinite number of $n$:
\begin{itemize}
\item[$\bullet$] $R_{X(n)}$ is not surjective,
\item[$\bullet$] $R_{X(n)}$ is not injective.
\end{itemize}
\noindent $\bullet$ $[R_{X(n)}$ not surjective$]$. There exists a vertex $v' \notin \textrm{Im}(R_{X(n)})$. Without loss of generality, we can assume that $|v'|<b$ where $b$ is the bound from the boundedness property of $F$. We will now consider a particular asymetric extension of $F(X(n))$, ${}^\Box F(X(n))$, where the chosen vertex in $F(X(n))$ is the furthest away from the pointed vertex $\varepsilon$. Indeed, if $F(X(n))$ is large enough, a vertex lying at maximal distance of $\varepsilon$ in $F(X(n))$ either has a free port or is part of a cycle, and thus is a valid vertex to perform the asymetric extension. Indeed, if this vertex has no free port, then any of its edge can be removed without splitting the graph, as it would contradict its maximality -- therefore it is in a cycle. Now, consider the graph $Y(n)=F^{-1}({}^\Box F(X(n)))$. Using uniform continuity of $F^{-1}$ and $R_\bullet$, and the fact that $|V(X(n))|$ is as big as we want, we have that there exists an index $n$ and a radius $r$ such that $Y(n)^r=X(n)^r$ and $R^b_{Y(n)^r}=R^b_{X(n)^r}$. 


Since $F(Y(n))$ is asymetric by construction, we have that $|F(Y)_{\diagup_\sim}|=|F(Y)|$. Applying lemma \ref{lem:invert2} gives us that $|Y_{\diagup_\sim}|=|F(Y)|$. In particular, this means that there are as many ways to place the pointer in $Y$ as there are vertices in $F(Y)$. Using the injectivity of $F$, we have that all vertices of $F(Y)$ are reached by $R_Y$, thus $v'\in \textrm{Im}(R_{Y(n)^r}^b)$. This contradicts $R^b_{Y(n)^r}=R^b_{X(n)^r}$. 

\noindent $\bullet$ $[R_{X(n)}$ not injective$]$. There exist two vertices $u,v \in X(n)$ such that $R_{X(n)}(u)=R_{X(n)}(v)$ and $u\neq v$. Without loss of generality, we can assume that $u =\varepsilon$ as $F$ is shift-invariant. According to lemma \ref{lem:invert1}, we have that $\varepsilon\sim_X v$. 


Moreover, applying the uniform continuity of $R_\bullet$ with $m=0$, we have that there exists a bound $l$ such that for all graph $X$ and all vertex $v$, $R_X(v)=\varepsilon$ implies $v\in X^l$.

Let us consider a asymetric extension of $X(n)$, ${}^\Box X(n)$, where the asymetric extension has been performed at maximal distance from $\varepsilon$, by the same argument as in the previous $\bullet$. In this graph, $\varepsilon$ and $v$ are not shift-equivalent and thus, $R_{{}^\Box X(n)}(\varepsilon)\neq R_{{}^\Box X(n)}(v)$.  By continuity of $R_\bullet$, we have that there exists a radius $r>l$ such that $R_{{}^\Box X(n)^r}^0=R_{ X(n)^r}^0$ for a large enough $n$, hence $R_{{}^\Box X(n)^r}^0(v)=R_{ X(n)^r}^0(v)=\varepsilon$, which contradicts $R_{{}^\Box X(n)}(\varepsilon)\neq R_{{}^\Box X(n)}(v)$.

~\\
\noindent [Infinite graphs]. Now we show that the result on finite graphs can be extended to infinite graphs, proving that for any infinite graph $R_X$ is bijective:

\noindent $\bullet$ $[R_X$ injective$]$. By contradiction. Take $X$ infinite such that there is $u\neq v$ and $R_X(u)=R_X(v)$. Without loss of generality we can take $u=\varepsilon$, i.e. $v\neq \varepsilon$ and $R_X(v)=\varepsilon$. By continuity of $R_\bullet$, there exists a radius $r$, which we can take larger than $|v|$ and $p$, such that $R_{X}=R_{X^r}$. Then $R_{X^r}(v)=R_X(v)=\varepsilon$, thus $R_{X^r}$ is not injective in spite of $X^r$ being finite and larger than $p$, leading to a contradiction. 

\noindent $\bullet$ $[R_X$ surjective$]$. By contradiction. Take $X$ infinite such that there is $v'$ in $F(X)$ and $v'\notin \textrm{Im}(R_X)$. By boundedness, there exists $u'\in F(X)$ such that $u'$ lies at distance less than $b$ of $v'$. Using shift-invariance, we can assume without loss of generality that $u'=\varepsilon$, hence, $|v'|<b$. By continuity of $R_\bullet$, there exists a radius $r$, which we can take larger than $p$, such that the images of $R_{X}$ and $R_{X^r}$ coincide over the disk of radius $b$. Then, $v'\notin \textrm{Im}(R_X)$ implies $v'\notin \textrm{Im}(R_{X^r})$, thus $R_{X^r}$ is not surjective in spite of $X^r$ being finite and larger than $p$, leading to a contradiction. \hfill $\Box$
\end{proof}
In \cite{ArrighiCGD,ArrighiIC}, Dowek and one of the authors reached an even more restrictive result: plain vertex-preservation. This is because their graphs are non-modulo, i.e. every vertex has a unique name. Thus being invertible in their setting meant being able to invert each of these names, which is more stringent. It was thus unclear to us whether vertex-preservation was a consequence of their stringent setting, or a more general fact. 

\section{Reversible Causal Graph Dynamics}\label{sec:reversibility}
A Reversible CGD (RCGD for short) is an invertible CGD whose inverse is also a CGD:

\begin{definition}[Reversible]
A CGD $(F,R_{\bullet})$ is {\em reversible} if there exists $S_\bullet$ such that $(F^{-1} ,S_{\bullet})$ is a CGD.
\end{definition}

Theorem \ref{th:preserv} shows that invertible  CGD are almost vertex-preserving. Notice that vertex-preservation guarantees that the inverse of a shift-invariant dynamics is a shift-invariant dynamics.


\begin{lemma}
If  $(F,R_\bullet)$ is an invertible, shift-invariant dynamics such that for all $X$, $R_X$ is a bijection, then $(F^{-1},S_\bullet)$ is a shift-invariant dynamics, with $S_{Y}=(R_{F^{-1}(Y)})^{-1}$.
\end{lemma}

\proof{
Consider $Y$ and $u'.v'\in Y$. Take $X$ and $u.v\in X$ such that $F(X)=Y$, $R_X(u)=u'$ and $R_X(u.v)=u'.v'$. We have:
${F^{-1}(Y_{u'})} ={F^{-1}(F(X)_{R_X(u)})}={F^{-1}(F(X_u))}={X_{(R_X)^{-1}(u')}}={F^{-1}(Y)_{S_Y(u')}}$. 
Moreover, take $v\in X_u$ such that $R_X(u.v)=R_X(u).R_{X_u}(v)=u'.v'$. We have: 
 $S_Y(u'.v') =(R_X)^{-1}(R_X(u.v)) =u.v =(R_{X})^{-1}(u').(R_{X_u})^{-1}(v')= S_Y(u').S_{Y_{u'}}(v')$.
~\qed}
 
We are now ready to prove our first main result, which states that the inverse of a causal graph dynamics is a causal graph dynamics:
\begin{theorem}[Invertible implies reversible]\label{th:main}
If $(F,R_\bullet)$ is an invertible CGD, then $(F,R_\bullet)$ is reversible.
\end{theorem}
\begin{proof}

We must construct $S_\bullet$. For $|F(X)|=|V(X)|>p$, we know that $R_X$ is bijective and we let $S_{F(X)}=R_X^{-1}$. For $|V(X)|\leq p$, we will proceed in two steps. First, we will construct an appropriate $S_{F(X)}$ for $X$. Second, we will make consistent choices for $S_{F(X)_{u'}}$ so that $S_\bullet$ is shift invariant. \\ 
We write $\tilde{u}$ for the shift-equivalence class of $u$ in $X$, and $\tilde{v'}$ for the shift-equivalence class of $v'$ in $F(X)$. For all $v'\in F(X)$, we make the arbitrary choice $S_{F(X)}(\tilde{v'})=v$, where $v$ is such that its image $R_X(v)$ is shift equivalent to $v'$ in $F(X)$, i.e.  $R_X(v)\sim_{F(X)} v'$. For this $X$, we have enforced $\sim$-compatibility. Then we make consistent choices for $S_{F(X)_{u'}}$. This is obtained by demanding that $S_{F(X)_{u'}}(\widetilde{\overline{u'}.v'})=\overline{u}.v$. Indeed, this accomplishes shift-invariance because $S_{F(X)_{u'}}(v')=S_{F(X)_{u'}}(\overline{u'}.u'.v')=\varepsilon.v'=v'$ implying the equality: $S_{F(X)}(u'.v')=u.v=S_{F(X)}(u').S_{F(X)_{u'}}(v')$. Moreover, $S_{F(X)_{u'}}$ is itself shift-invariant because: $S_{F(X)_{u'.v'}}(w')=S_{F(X)_{u'.v'}}(\overline{u'.v'}.u'.v'.w')=\overline{u.v}.u.v.w=w$, and $S_{F(X)_{u'}}(v')=v$ implying that $S_{F(X)_{u'}}(v'.w')=v.w=S_{F(X)_{u'}}(v').S_{F(X)_{u'.v'}}(w')$ , and $\sim$-compatible because $v'\sim w'$ implies  $S_{F(X)_{u'}}(v')=S_{F(X)_{u'}}(w')$, and thus $S_{F(X)_{u'}}(v')\sim S_{F(X)_{u'}}(w')$.\\ 
Continuity of the constructed $S_\bullet$ is due to the continuity of $R_\bullet$ and the finiteness of $p$.\\
Shift-invariance of $(F^{-1},S_\bullet)$ follows from $\sim$-compatibility of $S_\bullet$ and shift-inva\-riance of $(F,R_\bullet)$, because $F^{-1}(F(X)_u')=X_v$ where $v$ is such that $R_X(v)\sim u'$, hence $F^{-1}(F(X)_u')=X_{S_{F(X)}(u')}$. \qed

Continuity of $F^{-1}$ is directly given by the continuity of $F$ together with the compactness of $\mathcal{X}_{\Sigma,\Delta,\pi}$. Its boundedness derives either from the bijectivity of $R_X$ for $|V(X)|>p$ or from the finiteness of $X$ when $|V(X)|>p$.

\end{proof}

Notice that, ultimately, this result crucially relies on the compactness of $\mathcal{X}_{\Sigma,\Delta,\pi}$ which in turn relies on the boundedness of the degree $|\pi|$ and the finiteness of the internal states $\Sigma$ and $\Delta$. A contrario in \cite{ArrighiCGD,ArrighiIC}, Dowek and one of the authors reached a similar result, but with a much more lengthy proof, due to the lack of a compact metric space. Their setting was also more stringent, in the sense their graphs are non-modulo, and thus being invertible in their setting meant being able to invert each vertex name --- which is much to ask for. It was thus unclear to us whether reversibility was a consequence of their stringent setting, or a more general fact, that Theorem \ref{th:main} has now established.

\section{Block representation of reversible causal graph dynamics}\label{sec:blocks}

A famous result on RCA \cite{KariBlock}, is that these admit a finite-depth, reversible circuit form, with gates acting only locally. The result carries through to Quantum CA \cite{ArrighiUCAUSAL}. In order to apply these ideas to Reversible CGD, we must first make it clear what we mean by local operations. Afterwards, we will show that conjugating a local operation with an RCGD still yields a local operation, and deduce the block representation from this property.

\subsection{Locality}

Causal Graph Dynamics change the entire graph in one go---the word causal only referring to the constraint that information does not propagate too fast. Local operations, on the other hand, act just in one bounded region of the graph, leaving all of the rest unchanged. Here is the definition of being `local':
\begin{definition}[Local dynamics]\label{def:local}
A dynamics $(L,S_\bullet)$ is {\em $r'$-local} if it is uniformly continuous and bounded, and there exists $r'$ such that for all $X$ and $u'\in L(X)$ with $|u'|>r'$, there exists $u \in X$ such that we have both $L(X)^0_{u'}=X^0_u$ and for all $v\in X^0_u$,$\,S_X(u.v)=u'.v$.
\end{definition}

A local operation can also be shifted to act over the region surrounding some vertex $u$. The details of the next definition become apparent with Figure \ref{fig:shifteddyn}.
\begin{definition}[Shifted dynamics]\label{def:shifted}
Consider a dynamics $(L,S_\bullet)$ and some $u\in \Pi^*$. We define $L_u$ to be the map $X\mapsto(L(X_u))_{S_{X_u}(\overline{u})}$ if $u\in X$, and the identity otherwise. We define $S_{u,X}$ to be the map $v\mapsto \overline{S_{X_u}(\overline{u})}.S_{X_u}(\overline{u}.v)$ if $u\in X$, and the identity otherwise. We say that $(L_u,S_{u,\bullet})$ is $(L,S_\bullet)$ {\em shifted at $u$}.
\end{definition}

\begin{figure}[h!]
\begin{center}
\includegraphics[scale=0.8]{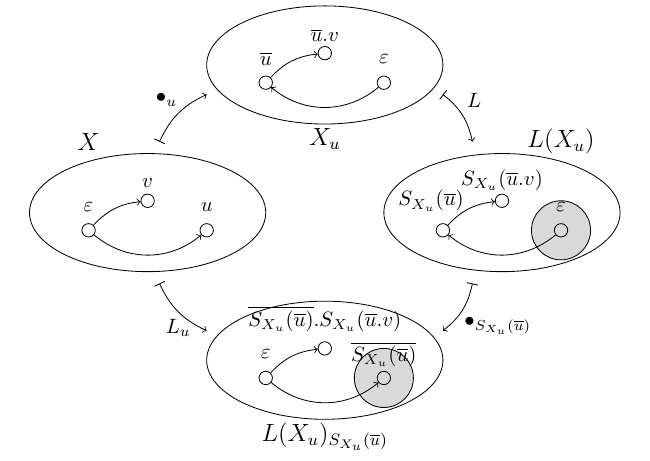}
\end{center}
\caption{\label{fig:shifteddyn}Shifted dynamics $L_u$. In the bottom graph $L_u(X)$, former vertex $v$ has name $\overline{S_{X_u}(\overline{u})}.S_{X_u}(\overline{u}.v)$.}
\end{figure}

The following two lemmas are technical and somewhat expected, but will turn out useful next.

\begin{lemma}[Bounded inflation]\label{lem:boundedinflation}
If $(L,S_\bullet)$ is $r'$-local, then for all $s$ there exists $s'$ such that for all $X$ and $v\in X$, if $|v|\leq s$, then $|S_X(v)|\leq s'$.
\end{lemma}
\begin{proof}
Suppose the contrary: there exists $s$ such that for all $s'$, $X(s')$ has some $|v(s')|\leq s$ such that $|S_X(v)|> s'$.
Since ${\cal X}_{\Sigma, \Delta, \ports}$ is compact \cite{ArrighiCayley}, $X(s')$ admits a subsequence which converges to some limit $X$, in the sense that $X(s'_k)^k=X^k$. For this particular $X$, for any $s'$, there is some $|v(s')|\leq s$ such that $|S_X(v)|> s'$. This is because we can choose $k$ so that $s'_k\geq s'$ and $k$ superior to the radius needed to determine $L(X)^{s'_k}=L(X(s'_k))^{s'_k}$, so that $|S_X(v)|=|S_{X^k}(v)|=|S_{{X(s_k)}^k}(v)|=|S_{{X(s_k)}}(v)|>s'_k\geq s'$. Thus, there exists a point of $v\in X$ such that  $|S_X(v)|> \infty$, which is a contradiction.
\end{proof}

\begin{lemma}\label{lem:tothet}
If $(L,S_\bullet)$ is $r'$-local, for all $t$, for all $u'\in L(X)$ with $|u'|>r'+t+1$, there exists $u\in X$ with $S_X(u)=u'$ such that we have:
\begin{itemize}
\item[$\bullet$] $L(X)^t_{u'}=X^t_u$,
\item[$\bullet$] $\forall v\in X^t_u,\,S_X(u.v)=u'.v$.
\end{itemize}
\end{lemma}

\begin{proof} 
Take such a $u'$ and consider $u$ such that $u'=S_X(u)$. \\
$[\textrm{First }\bullet]$ Since $|u'|>r'+t+1$, we have that for all $v\in L(X)^t_{u'}$, $|u'.v|>r'$. Hence, by $r'$-locality of $L$, there exists $x\in X$ such that $S_X(x)=u'.v$ and such that $L(X)_{u'.v}^0=X_x^0$, i.e. the vertex $v$ in $L(X)^t_{u'}$, in terms of its internal states and edges, is the same as the vertex $x$ in $X$. Now, say there exists $|z|=1$ such that $w=v.z\in L(X)^t_{u'}$, i.e. there is an edge between $v$ and $v.z$ in $L(X)^t_{u'}$. Again since $|S_X(x)|>r'$, the $r'$-locality yields $u'.v.z=S_X(x).z=S_X(x.z)$, i.e. the edge between $v$ and $v.z$  $v$ in $L(X)^t_{u'}$ is the same as that between $x$ and $x.z$ in $X$. Consider $v_1\ldots v_k=v$ with $k\leq t$ and $|v_i|=1$. A similar argument starting from $u'$ and following these edges shows that $x$ is at distance $t$ of $u$ in $X$, and thus $x.z$ is at distance $t+1$ of $u$ in $X$. So the vertices $x$, $x.v$ and their edge do appear in $X_u^t$.\\
$[\textrm{Second }\bullet]$ Again take $w\in X^t_u = L(X)^t_{u'}$. Consider $w_1\ldots w_k=w$ with $k\leq t+1$ and $|w_i|=1$. Since $|u'|>r'+t+1>r'$, the $r'$-locality applies and yields $S_X(u.w_1)=S_X(u).w_1=u'.w_1$. Similarly, since $|u'.w_1\ldots w_i|>r'+t+1-i>r'$, the $r'$-locality applies and yields $S_X(u.w_1\ldots w_i.w_{i+1})=S_X(u.w_1\ldots w_i).w_{i+1}=u'.w_1\ldots w_i.w_{i+1}$. Eventually $S_X(u.w)=u'.w$.\qed
\end{proof}

We may wish to apply a series of local operations at different positions $u_i$, i.e. a circuit. However, applying a local operation may change the graph and hence vertex names, hence some care must be taken:
\begin{definition}[Product]\label{def:product}
Given a shifted dynamics $(L_u,S_{u,\bullet})$ and a dynamics $(M,T_\bullet)$, it is convenient to define their composition so as to take care of renamings 
$$(L_{u} \cdot M)(X)=(L_{T_X(u)})(M(X))$$
and similarly for the $S$ maps:
$$(S_{u,\bullet} \cdot T_{\bullet})_X(v)=(S_{T_X(u),M(X)})(T_{X}(v)).$$
Then,
$$(\prod_{[u_1,\ldots,u_n]} L)(X)=\big((\prod_{[u_2,\ldots,u_n]} L)\cdot L_{u_1}\big)(X)$$
and similarly for the $S$ maps. In order to extend to infinite products, we denote by $u^r$ the restriction of a series $u:\N\rightarrow V(X)$ to the codomain $V(X^r)$, and then take the following limit:
$$(\prod_u L)(X)= \lim_{r\rightarrow\infty} \quad (\prod_{u^r} L) (X),$$
and similarly for the $S$ maps.
When the $L_u$ commute, only the codomain of $u$ matters. In this case, products over sets of vertices are well-defined, e.g. $(\prod_W L)$.
\end{definition}


\subsection{Block representation}

A famous result about RCA \cite{KariBlock}, is that they admit a finite-depth, reversible circuit form, with gates acting only locally. The result carries through to Quantum CA \cite{ArrighiUCAUSAL}. A crucial step towards this result is to show that conjugating a local operation with an RCGD still yields a local operation:
\begin{proposition}\label{prop:conjloc}
If $(F,R_\bullet)$ is an RCGD and  $(L,S_\bullet)$ is a local dynamics, then $(L',T_\bullet)$ is a local dynamics, with 
\begin{itemize}
\item[$\bullet$] $L'=F^{-1}\circ L\circ F$ and 
\item[$\bullet$] $T_X(u)=R'_{F^{-1}(L(F(X))}(S_{F(X)}(R_X(u)))$,
\end{itemize} where the function $R'_\bullet$ is such that $(F^{-1},R'_\bullet)$ is a CGD. 
\end{proposition}

\begin{proof} 
Boundedness and uniform continuity by composition. 
Next, suppose:
$L$ is local,  $r_0$ is such that for all $X,Y$ if $X^{r_0}=Y^{r_0}$ then $F^{-1}(X)^0=F^{-1}(Y)^0$ (given by uniform continuity of $F^{-1}$), $r_{2b_F}$ is such that for all $X,Y$ if $X^{r_{2b_F}}=Y^{r_{2b_F}}$ then $F(X)^{2b_F}=F(Y)^{2b_F}$ (given by uniform continuity of $F$), $b_{F^{-1}}$ is the bound given by the bounded inflation lemma applied on $F^{-1}$, $b_L$ is the bound given by the boundedness of $L$ and $r_L$ the radius of locality of $L$. In the  two following points, we chose a radius $r'$ as follow:
$$r'=b^{F^{-1}}(r_L+2+max(r_0,2b_F,r_{2b_F}))$$
Consider $|u'|> r'$.\\
\noindent$[\textrm{First }\bullet]$ Let us show that there exists $u\in X$ such that $L'(X)^0_{u'} = X^0_u$.
 By definition of $F^{-1}$, there exists $w\in LF(X)$ such that $R^{F^{-1}}_{LF(X)}(w)=u'$. By bounded inflation of $F^{-1}$, we have $|w| >r_L$ and thus by locality of $L$, there exists $w'\in F(X)$ such that $S_{F(X)}(w')=w$. Finally by reversibility of $F$ there exists $u\in X$ such that $R^F_X(u)=w$, and thus $u'=T_X(u)$.
Notice that we have that $|S_{F(X)}R^F_{X}(u)|> r_0+r_L+2$. Using lemma \ref{lem:tothet} with $t=r_0$, we have: $LF(X)^{r_0}_{S_{F(x)}R^F_X(u)}=$ $F(X)^{r_0}_{R^F_X(u)}$ $= F(X_u)^{r_0}$. By definition of $r_0$, $F(X_u)^{r_0} = LF(X)^{r_0}_{S_{F(x)}R^F_X(u)}$ implies $X_u^0 = $  $F^{-1}(LF(X)^{r_0}_{S_{F(x)}R^F_X(u)})^0$, which leads by shift-invariance of $F^{-1}$ to $X_u^0 ~=~$ $ F^{-1}(LF(X)^{r_0})_{R^{F^{-1}}_{LF(X)}S_{F(x)}R^F_X(u)}^0$.  Hence $X^0_u=L'(X)^0_{u'}$.
~\\
$[\textrm{Second }\bullet]$
Consider $u$ as above and $v\in X^0_u$.
\begin{align*}
T_X(u.v) & =R^{F'}_{LF(X)}S_{F(x)}(R^F_X(u.v))\\
& =R^{F'}_{LF(X)}S_{F(x)}(R^F_X(u).R^F_{X_u}(v))\textrm{ using shift-invariance of $F$ }\\
& =R^{F'}_{LF(X)}S_{F(x)}(R^F_X(u)).R^F_{X_u}(v) \textrm{ because }|S_{F(X)}R^{F}_X(u)|> r_L +2b^F+2\\
& =R^{F'}_{LF(X)}(S_{F(x)}(R^F_X(u))).R^{F'}_{LF(X)_{S_{F(X)}(R^F_X(u))}}(R^F_{X_u}(v))\textrm{ using shift-invariance of $F^{-1}$ }\\
& =T_X(u).R^{F'}_{LF(X)_{S_{F(X)}(R^F_X(u))}}(R^F_{X_u}(v))
\end{align*}
\noindent We will now show that: $v=R^{F'}_{LF(X)_{S_{F(X)}(R^F_X(u))}}(R^F_{X_u}(v))$.
Since $|R^F_{X_u}(v)| < 2b_F$ by bounded inflation of $F$, it is enough to show: $${R^ {F'}_{LF(X)_{S_{F(X)}(R^F_X(u))}} }^{2b_F}(R^F_{X_u}(v))=v.$$
By definition of $r_{2b_F}$, we have that: if $X^{r_{2b_F}}=Y^{r_{2b_F}}$ then $R^{F'}(X)^{2b_F}=R^{F'}(Y)^{2b_F}$. Let us show that $LF(X)_{S_{F(X)}^{2b_F}(R^F_X(u))}^{r_{2b_F}}= F(X_u)^{r_{2b_F}}$. By applying lemma \ref{lem:tothet} with $t=r_{2b_F}$,
$LF(X)_{S_{F(X)}^{2b_F}(R^F_X(u))}^{r_{2b_F}}=F(X)_{R^F_X(u)}^{r_{2b_F}}$ which, by shift-invariance of $F$, is equal to $F(X_u)^{r_{2b_F}}$. As a consequence,  
$${  R^{F'}_{LF(X)_{S_{F(X)}(R^F_X(u))}}  }
^{2b_F}(R^F_{X_u}(v))
={R^{F'}_{F(X_u)}}^{2b_F}(R^F_{X_u}(v))
=R^{F'}_{F(X_u)}(R^F_{X_u}(v))
=v$$ by definition of $R^{F'}_\bullet$ ~\qed
\end{proof}

Second, we give ourselves a little more space so as to mark which parts of the graph have been updated, or not. 

\begin{definition}[Marked pointed graph modulo]
Consider the set of pointed graphs modulo $\mathcal{X}_{\Sigma,\Delta,\pi}$ with labels in $\Sigma$, and ports in $\pi$. Let $\Sigma'=\Sigma\times\{0,1\}$ and $\pi'=\pi \times \{0,1\}$. We define the set of {\em marked} pointed graphs modulo $\overline{\mathcal{X}}_{\Sigma',\Delta,\pi'}$ to be the subset of ${\mathcal{X}}_{\Sigma',\Delta,\pi'}$ such that for all $u\in X$ having label $(x,a)$, and edge $\{u\port(i,b),v\port (j,c)\}\in X$, we have $a=c$. Given a graph $X\in \mathcal{X}_{\Sigma,\Delta,\pi}$, it is naturally identified with the same graph in $\overline{\mathcal{X}}_{\Sigma',\Delta,\pi'}$ with all marks set to $0$.
\end{definition}

\begin{definition}[Mark operation]
We define the marking operation $\mu(.)$ over labels and ports as toggling the bit in the second component:
\begin{itemize}
\item $\forall (x,a)\in \Sigma', \mu(x,a)=(x,1-a)$
\item $\forall (i,a)\in \pi', \mu(i,a)=(i,1-a)$
\end{itemize}
Then, we define the mark operation $\mu(.)$ over pointed graphs modulo, as attempting to mark the pointed vertex label and its opposite ports, if this will not create conflicts between ports. More formally, given a graph $X$ in $\overline{\mathcal{X}}_{\Sigma',\Delta,\pi'}$, we define the mark operation, $\mu:\overline{\mathcal{X}}_{\Sigma',\Delta,\pi'}\rightarrow \overline{\mathcal{X}}_{\Sigma',\Delta,\pi'}$ as follows:
\begin{itemize}
\item if $\exists v,w\in X ,i,j\in \pi'$ such that $\{\varepsilon\port i ,  v \port j\}\in X$ and $\{ v\port \mu(j) ,w \port k\}\in X$ then $\mu(X)=X$
\item else
\begin{itemize}
\item[$\bullet$] $\sigma_{\mu(X)}(\varepsilon)=\mu(\sigma_{X}(\varepsilon))$
\item[$\bullet$] For all $i,j\in\pi'$,  $\{\varepsilon\port \mu(i),\varepsilon\port \mu(j)\} \in \mu(X)$ if $\{\varepsilon\port i,\varepsilon\port j\} \in X$.
\item[$\bullet$] For all $v\in X$ with $v\neq \varepsilon$ and $i,j\in\pi'$,  $\{\varepsilon\port i,v\port \mu(j)\} \in \mu(X)$ if $\{\varepsilon\port i,v\port j\} \in X$.
\end{itemize}
the rest of the graph $X$ is left unchanged.
\end{itemize}
\end{definition}

Remark: The set $\overline{\mathcal{X}}_{\Sigma',\Delta,\pi'}$ is a compact subset of $\mathcal{X}_{\Sigma',\Delta,\pi'}$.

$\overline{\mathcal{X}}_{\Sigma',\Delta,\pi'}$ allows for a clear distinction between marked and unmarked parts of the graph, allowing to extend any RCGD $F$ over $\mathcal{X}_{\Sigma,\Delta,\pi}$ to act over $\overline{\mathcal{X}}_{\Sigma',\Delta,\pi'}$, as usual on the unmarked part, and trivially on the marked part.

\begin{definition}[Reversible extension]
Let $F:\mathcal{X}_{\Sigma,\Delta,\pi}\rightarrow \mathcal{X}_{\Sigma,\Delta,\pi}$ be an RCGD. We say that $F':\overline{\mathcal{X}}_{\Sigma',\Delta,\pi'}\rightarrow\overline{\mathcal{X}}_{\Sigma',\Delta,\pi'}$ is a {\em reversible extension} of $F$ if and only if $F'$ is an RCGD, and:
\begin{itemize}
\item[$\bullet$] For all $X \in \mathcal{X}_{\Sigma\times\{0\},\pi\times\{0\}}$, $F'(X)=F(X)$.
\item[$\bullet$] For all $X \in \mathcal{X}_{\Sigma\times\{1\},\pi\times\{1\},}$, $F'(X)=X$.
\item[$\bullet$] For all $|V(X)|\leq p$ and $X \notin \mathcal{X}_{\Sigma\times\{0\},\pi\times\{0\}}$, $F'(X)=X$, where $p$ is that of Theorem \ref{th:preserv}.
\end{itemize}
\end{definition}

\begin{proposition}[Reversible extension]\label{prop:revextension}
Suppose $F:\mathcal{X}_{\Sigma,\Delta,\pi}\rightarrow \mathcal{X}_{\Sigma,\Delta,\pi}$ is an RCGD. Then it admits a reversible extension $F':\overline{\mathcal{X}}_{\Sigma',\Delta,\pi'}\rightarrow\overline{\mathcal{X}}_{\Sigma',\Delta,\pi'}$.
\end{proposition}

\noindent{$[$\bf Proof of Proposition \ref{prop:revextension}$]$}\label{app:prop:revextension}

\noindent{\em Proof outline.} In order to define a reversible extension $F'$, we first separate the graph in two layers: the marked and unmarked layers. We then isolate the connected components in each layer and express them as ``shifted'' graphs modulo. Once this is done, it is easy to define $F'$ as acting as $F$ on the unmarked components and as the identity on the other components.

\noindent\emph{Construction of a reversible extension.}

We first isolate the two layers inside our graphs:

\begin{definition}[Upper and lower projections]
Let $G$ be a graph in $G(\overline{\mathcal{X}}_{\Sigma',\Delta,\pi'})$. We define ${}^\downarrow\! G$ (resp. ${}^\uparrow\! G$) the lower (resp. upper) projection of $G$ as the set of the connected component obtained after removing all marked vertices (resp. all non-marked vertices without used marked ports).
\end{definition}

We then add some structure on those components:

\begin{lemma}[Characterization of connected components]
Given $G$ in $G(\overline{\mathcal{X}}_{\Sigma',\Delta,\pi'})$, the elements of the sets ${}^\downarrow\! G$ and ${}^\uparrow\! G$ are of the form $u.Y$ with $u\in L(X)$ and $Y\in G(\overline{\mathcal{X}}_{\Sigma',\Delta,\pi'})$.
\end{lemma}
\begin{proof}
Consider an element $H$ of ${}^\downarrow\! G$. In particular, it is a graph of $\overline{\mathcal{G}}_{\Sigma',\Delta,\pi'}$. Consider any vertex $u$ of $H$. The graph $\overline{u}.H$ has a vertex having $\varepsilon$ as name. We can now construct the graph modulo $X=\tili{\overline{u}.H} \in \overline{\mathcal{X}}_{\Sigma',\Delta,\pi'}$. By construction this graph verifies $u.G(X)=H$.
By symmetry, the same holds for components in ${}^\uparrow\! G$.\qed

\end{proof}

\begin{proposition}[Reversible extension]
Any RCGD  $(F,R_\bullet)$ over $\mathcal{X}_{\Sigma,\Delta,\pi}$ admits a reversible extension $(F',R'_{\bullet})$ over $\overline{\mathcal{X}}_{\Sigma',\Delta,\pi'}$.
\end{proposition}

\begin{proof}
Let us construct such a reversible extension $F'$. Let $p$ be that of Theorem \ref{th:preserv}. For all $|V(X)|\leq p$ and $X \notin \mathcal{X}_{\Sigma\times\{0\},\Delta,\pi\times\{0\}}$, we let $F'(X)=X$. The rest supposes $|V(X)|>p$.\\
Given $L:\mathcal{X}_{\Sigma,\Delta,\pi}\rightarrow \mathcal{X}_{\Sigma,\Delta, \pi}$, we define $L^\star:\mathcal{X}_{\Sigma,\Delta, \pi} \rightarrow \mathcal{G}_{\Sigma,\Delta, \pi}$ as the function $G\circ L$. 
Now for all $X\in \overline{\mathcal{X}}_{\Sigma',\Delta,\pi'}$, we define $F'(X)$ as the equivalence class modulo isomorphism of the following graph pointed on $\varepsilon$:
$$ \left[ \bigcup_{C\in {}^\uparrow\! G} C \right]  \cup \left[\bigcup_{u.Y\in {}^\downarrow\! G}  L_0(u.F^\star(P(Y)))  \right] $$
where $G=G(X)$, $L_0: \mathcal{G}_{\Sigma,\Delta, \pi}\rightarrow \overline{\mathcal{G}}_{\Sigma',\Delta,\pi'}$ is the map that adds $0$ in the label of each vertex, and $P:\overline{\mathcal{X}}_{\Sigma',\Delta,\pi'} \rightarrow \mathcal{X}_{\Sigma,\Delta,\pi}$ is maps that forgets about the markings.

Notice that if $G\in \mathcal{G}_{\Sigma\times\{0\},\pi\times\{0\}}$ then ${}^\uparrow\! G$ is empty and ${}^\downarrow\! G$ contains a single connected component $\varepsilon.G$ (the graph itself), thus $F'$ computes $F$. On the other hand, if 
$G \in \mathcal{G}_{\Sigma\times\{1\},\pi\times\{1\}}$ then ${}^\downarrow\! G$ is empty and ${}^\uparrow\! G$ contains $\varepsilon.G$ only, thus $F'$ computes the identity. Hence this $F'$ is a good candidate for being a reversible extension of $F$. It remains now to check that $F'$ is causal, vertex-preserving and reversible.\\
$[$\emph{Causal}$]$ 
Shift-invariance, boundedness and continuity follow directly from the shift-invariance, boundedness and continuity of both $F$ and the identity.
\noindent\\
$[$\emph{Reversible}$]$ Replace $F$ by $F^{-1}$ in the previous definition.\qed
\end{proof}

In order to obtain our circuit-like form for RCGD, we will proceed by reversible, local updates. 

\begin{definition}[Conjugate mark]
Given a reversible extension $F':\overline{\mathcal{X}}_{\Sigma',\Delta,\pi'}\rightarrow\overline{\mathcal{X}}_{\Sigma',\Delta,\pi'}$, we define the {\em conjugate mark} $K:\overline{\mathcal{X}}_{\Sigma',\Delta,\pi'}\rightarrow\overline{\mathcal{X}}_{\Sigma',\Delta,\pi'}$ to be the function:
$$K= F'^{-1} \circ \mu \circ F'.$$
\end{definition}
Notice that by Proposition \ref{prop:conjloc}, the local update blocks are local operations. Moreover, since they are defined as a composition of invertible dynamics, so they are. In order to represent the whole of an RCGD, it suffices to apply these local update blocks at every vertex.
\begin{theorem}[Reversible localizability]\label{th:revloc}
Suppose $F:\mathcal{X}_{\Sigma,\Delta,\pi}\rightarrow \mathcal{X}_{\Sigma,\Delta,\pi}$ is an RCGD. Then for all $X\in \mathcal{X}_{\pi,\Sigma,\Delta}$, we have that:
$$F(X) = (\prod_X \mu)\cdot(\prod_X K)(X)$$ 
where $K=F'^{-1}\circ \mu \circ F'$ for $F':\overline{\mathcal{X}}_{\Sigma',\Delta,\pi'}\rightarrow\overline{\mathcal{X}}_{\Sigma',\Delta,\pi'}$ a reversible extension of $F$.
\end{theorem}
\begin{proof} 
 Let us consider a graph $X\in \mathcal{X}_{\Sigma,\Delta,\pi}$. We have:
\begin{align*}
  (\prod_X \mu)\cdot(\prod_X K)(X)
& = (\prod_X \mu)\cdot(\prod_X F'^{-1}\mu F')(X) \\
& = (\prod_X \mu)\cdot\left(F'^{-1}(\prod_X \mu) F'\right)(X) & (1)\\
& = (\prod_X \mu)\cdot(\prod_X \mu ) F(X)& (2)\\
& = F(X).
\end{align*}
The argument to go from $(1)$ to $(2)$ depends on $|X|$:
\begin{itemize}
\item If $|X|\leq p$ (i.e $F$ might change the size of $X$), then all the vertices in $(\prod_X \mu) F'(X)$ are not necessarily marked, but by definition of $F'$ (third $\bullet$), $F'^{-1}$ will perform the identity over this graph.
\item If $|X|>p$, then there is no ambiguity and $(\prod_X \mu) F'(X)$ is simply the graph $F(X)$ where all vertices have been marked. Indeed, since $X$ is in $\mathcal{X}_{\Sigma\times\{0\},\Delta,\pi\times\{0\}}$, so is $F'(X)$, and so $(\prod_X \mu)F'(x)$ is in $\mathcal{X}_{\Sigma\times\{1\},\Delta,\pi\times\{1\}}$.
\end{itemize}


\end{proof}
~\\
Notice that the cases $|X|\leq p$ are finite and $F$ is bijective, thus it just permutes those cases. Thus, this theorem generalizes the block decomposition of reversible cellular automata, which represents any reversible cellular automata as a circuit of finite depth of local permutations. Here, the mark $\mu$ and its conjugate $K$ are the local permutations. The circuit is again of finite depth, a vertex $u$ will be attained by all those $K$ that act over $X_u^{r'}$, where $r'$ is the locality radius of $K$. Therefore, the depth is less than $|\pi|^{r'}$. An example of such a decomposition is described in Figure \ref{fig:movingheadblock}.

\begin{figure}[h!]
\begin{center}
\begin{tabular}{ccc}
(1)\includegraphics[scale=0.26]{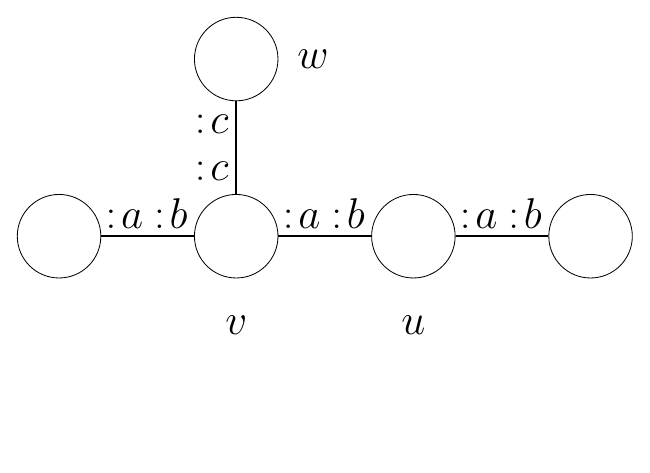} &
(2)\includegraphics[scale=0.26]{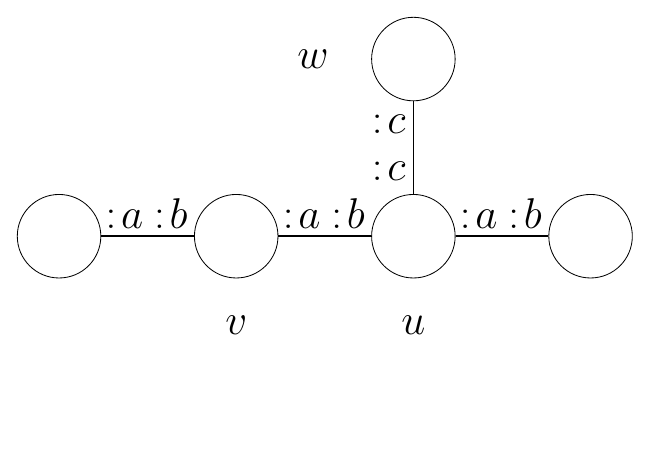} &
(3)\includegraphics[scale=0.26]{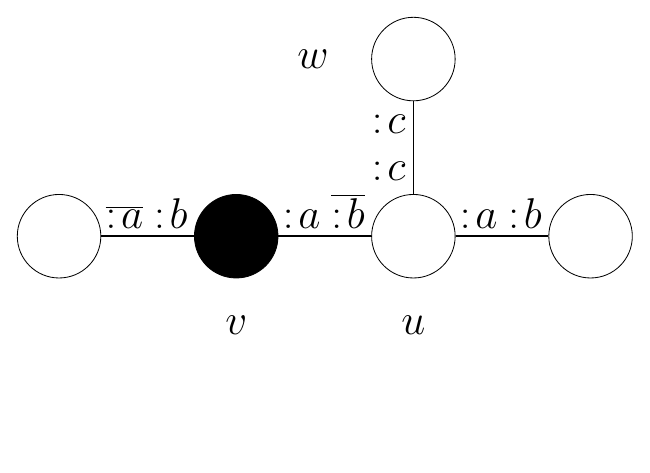} \\
(4)\includegraphics[scale=0.26]{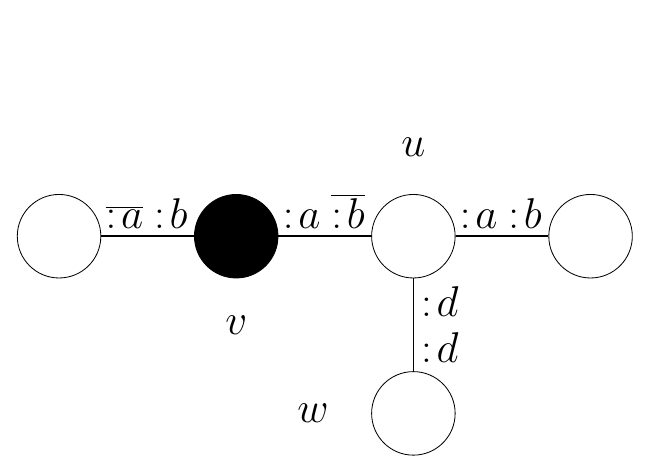} &
(5)\includegraphics[scale=0.26]{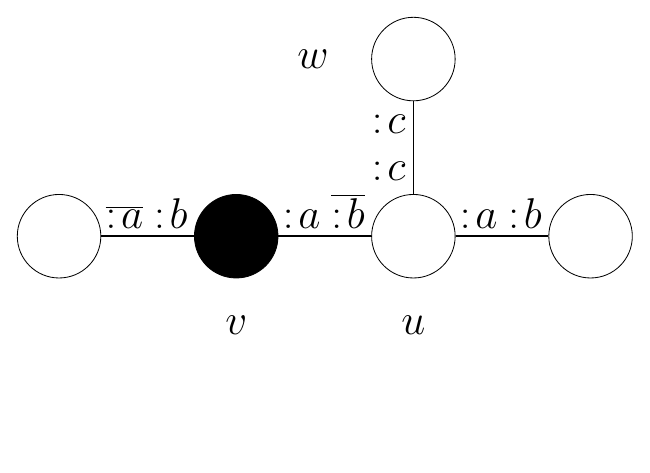} &
(6)\includegraphics[scale=0.26]{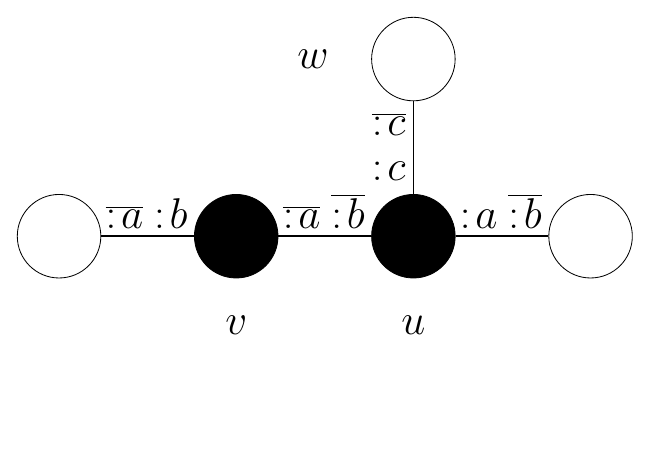} \\
(7)\includegraphics[scale=0.26]{Pictures/exemple_step12_bw.pdf} &
(8)\includegraphics[scale=0.26]{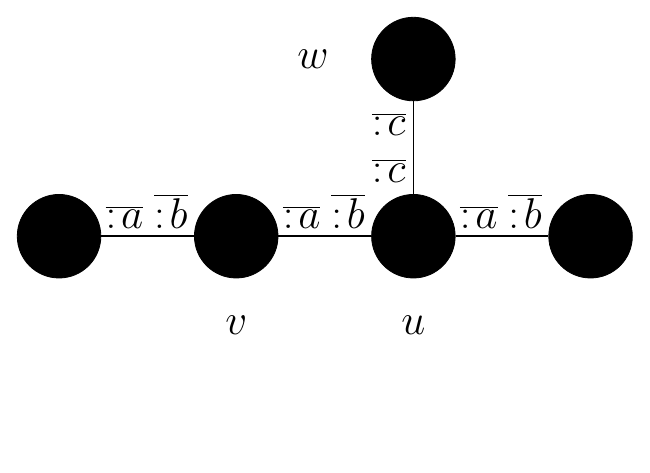} &
(9)\includegraphics[scale=0.26]{Pictures/exemple_step3.pdf} \\
\end{tabular}
\end{center}
\caption{\label{fig:movingheadblock}Block representation of the moving head dynamics. $(1)$ Initially, no vertices are marked. $(2)$ to $(4)$ Application of $K_v$. First $F$ is applied, then $v$ is marked, followed by the application of $F^{-1}$. $(5)$ to $(7)$ Application of $K_u$. $(8)$ The graph once every $K$ have been applied. The vertices just need to be unmarked by the $\mu$'s. $(9)$ Altogether this implements one time step of F.}

\end{figure}

\section{Conclusion}

{\em Summary of results.} We have studied Reversible Causal Graph Dynamics, thereby extending Reversible Cellular Automata results to time-varying, pointed graphs modulo. Pointed graphs modulo are arbitrary bounded-degree networks, with a pointed vertex serving as the origin, and modulo renaming of vertices. Some of these graphs have shift-equivalent vertices. We have shown that if a Causal Graph Dynamics (CGD) is invertible, then it preserves shift-equivalence classes. This in turn entails almost-vertex-preservation, i.e. the conservation of each vertex for big enough graphs. Next, we have shown that the inverse of a CGD is a CGD.
Finally, we have proved that Reversible Causal Graph Dynamics can be represented as finite-depth circuits of local reversible gates.

{\em Future work.} 
We have shown that invertible causal graph dynamics implies almost vertex-preservation or, in other words, that beyond some finitary cases, information conservation implies conservation of the systems that support this information. The result could perhaps be understood as a ``Matter conservation theorem'', \`a la Lavoisier. Still, this cannot forbid that some `dark matter' which was there at all times, could now be made `visible'. We plan to follow this idea in a subsequent work. We also wish to explore the quantum regime of these models, as similar results where given for Quantum Cellular Automata over fixed graphs \cite{ArrighiUCAUSAL}. Such results would be of interest to theoretical physics, in the sense of discrete time versions of \cite{QuantumGraphity1}.

\section*{Acknowledgements} This work has been funded by the ANR-12-BS02-007-01 TARMAC grant, the ANR-10-JCJC-0208 CausaQ grant, and the John Templeton Foundation, grant ID 15619. The authors acknowledge enlightening discussions with Bruno Martin and Emmanuel Jeandel. This work has been partially done when PA was delegated at Inria Nancy Grand Est, in the project team Carte. 

\bibliography{biblio}
\bibliographystyle{plain}

\appendix

\section{Formalism}\label{app:formalism}

This appendix provides formal definitions of the kinds of graphs we are using, together with the operations we perform upon them. None of this is specicific to the reversible case; it can all be found in \cite{ArrighiCayleyNesme} and is reproduced here only for convenience.

\subsection{Graphs}\label{app:graphs}

\noindent {\em Vertex names.} Let $\pi$ be a finite set, $\Pi=\pi^2$, and $V={\cal P}(\Pi^*)$, where `$.$' represents concatenation of words and $\varepsilon$ is the empty word, as usual. Each vertex of a graphs non-modulo will be uniquely identified with a name $u$ in $V$. This particular choice of the universe of names is irrelevant until Definition \ref{def:associatedgraph}, when it becomes natural.

\begin{definition}[Graph non-modulo]\label{def:graphs}
A {\em graph non-modulo} $G$ is given by 
\begin{itemize}
\item[$\bullet$] An at most countable subset $\vt(G)$ of $V$, whose elements are called {\em vertices}.
\item[$\bullet$] A finite set $\ports$, whose elements are called {\em ports}.
\item[$\bullet$] A set $E(G)$ of non-intersecting two element subsets of $\vt(G)\port\ports$, whose elements are called edges. In other words an edge $e$ is of the form $\{u \port a, v \port b\}$, and $\forall e,e'\in E(G), e\cap e'\neq \emptyset \Rightarrow e=e'$. 
\end{itemize}
The graph is assumed to be connected: for any two $u,v\in V(G)$, there exists $v_0,\ldots , v_{n}\in V(G)$, $a_0,b_0\ldots , a_{n-1},b_{n-1}\in \pi$ such that for all $i\in\{0\ldots n-1\}$, one has $\{v_i\port a_i,v_{i+1}\port b_i\}\in E(G)$ with $v_0=u$ and $v_n=v$.
\end{definition}

\begin{definition}[Labelled graph non-modulo]
A labelled graph non-modulo is a triple $(G,\sigma,\delta)$, also denoted simply $G$ when it is unambiguous, where $G$ is a graph, and $\sigma$ and $\delta$ respectively label the vertices and the edges of $G$:
\begin{itemize}
\item[$\bullet$] $\sigma$ is a partial function from $V(G)$ to a finite set $\Sigma$;
\item[$\bullet$] $\delta$ is a partial function from $E(G)$ to a finite set $\Delta$.
\end{itemize}
The {\em set of all graphs} with ports $\ports$ is written ${\cal G}_{\ports}$.  
The {\em set of labelled graphs} with states $\Sigma,\Delta$ and ports $\ports$ is written ${\cal G}_{\Sigma,\Delta,\ports}$.
To ease notations, we sometimes write $v \in G$ for $v \in \vt(G)$.
\end{definition}

We single out a vertex as the origin:
\begin{definition}[Pointed graph non-modulo]\label{def:pointedgraph}
A {\em pointed (labelled) graph} is a pair $(G,p)$ with $p\in G$. 
The {\em set of pointed graphs}  with ports $\ports$  is written ${\cal P}_{\ports}$.
The {\em set of pointed labelled graphs} with states $\Sigma,\Delta$ and ports $\ports$ is written ${\cal P}_{\Sigma,\Delta,\ports}$.
\end{definition}

Here is when graph differ only up to renaming:
\begin{definition}[Isomorphism]\label{def:isomorphism}
An {\em isomorphism} $R$ is a function from ${\cal G}_{\pi}$ to ${\cal G}_{\pi}$ which is specified by a bijection $R(.)$ from $V$ to $V$. 
The image of a graph $G$ under the isomorphism $R$ is a graph $RG$ whose set of vertices is $R(\vt(G))$, and whose set of edges is $\{\{R(u):a,R(v):b\} \;|\; \{u:a,v:b\}\in E(G) \}$. 
Similarly, the image of a pointed graph $P=(G,p)$ is the pointed graph $RP=(RG,R(p))$. 
When $P$ and $Q$ are isomorphic we write $P\approx Q$, defining an equivalence relation on the set of pointed graphs. The definition extends to pointed labelled graphs.
\end{definition}
\noindent (Pointed graph isomorphism rename the pointer in the same way as it renames the vertex upon which it points; which effectively means that the pointer does not move.)

\noindent Our main objects of study are pointed graphs modulo isomorphism. 
\begin{definition}[pointed graphs modulo]\label{def:pointedmodulo}
Let $P$ be a pointed (labelled) graph $(G,p)$. The {\em pointed graph modulo} $X$ is  $\widetilde{P}$ the equivalence class of $P$ with respect to the equivalence relation $\approx$. The {\em set of pointed graphs modulo} with ports $\ports$ is written ${\cal X}_{\ports}$. The {\em set of labelled pointed Graphs modulo} with states $\Sigma,\Delta$ and ports $\ports$ is written ${\cal X}_{\Sigma,\Delta,\ports}$.
\end{definition}

\subsection{Paths and vertices}\label{app:paths}

When we are considering pointed graphs modulo isomorphism, vertices no longer have a unique identifier. Still they can be designated by a sequence of ports in $(\pi^2)^*$ that leads, from the origin, to this vertex.

\begin{definition}[Path]\label{def:path}
Given a pointed graph modulo $X$, we say that $\alpha\in\Pi^*$ is a path of $X$ if and only if there is a finite sequence $\alpha=(a_i b_i)_{i\in\{0,...,n-1\}}$ of ports such that, starting from the pointer, it is possible to travel in the graph according to this sequence. More formally, $\alpha$ is a path if and only if there exists $(G,p)\in X$ and there also exists $v_0,\ldots , v_{n}\in V(G)$ such that for all $i\in\{0\ldots n-1\}$, one has $\{v_i\port a_i,v_{i+1}\port b_i\}\in E(G)$, with $v_0=p$ and $\alpha_i=a_ib_i$. Notice that the existence of a path does not depend on the choice of $(G,p)\in X$. The set of paths of $X$ is denoted by $L(X)$. 
\end{definition}
Notice that paths can be seen as words on the alphabet $\Pi$ and thus come with a natural operation `$.$' of concatenation, a unit $\varepsilon$ denoting the empty path, and a notion of inverse path $\overline{\alpha}$ which stands for the path $\alpha$ read backwards. Two paths are equivalent if they lead to same vertex:
\begin{definition}[Equivalence of paths]
Given a pointed graph modulo $X$, we define the {\em equivalence of paths} relation $\equiv_{X}$ on $L(X)$ such that for all paths $\alpha,\alpha'\in L(X)$, $\alpha\equiv_{X} \alpha'$ if and only if, starting from the pointer, $\alpha$ and $\alpha'$ lead to the same vertex of $X$. 
More formally, $\alpha\equiv_{X}\alpha'$ if and only if there exists $(G,p)\in X$ and $v_1,\ldots , v_{n},v'_1,\ldots , v'_{n'}\in V(G)$ such that for all $i\in\{0\ldots n-1\}$, $i'\in\{0\ldots n'-1\}$, one has $\{v_i\port a_i,v_{i+1}\port b_i\}\in E(G)$, $\{v'_{i'}\port a'_{i'},v'_{i'+1}\port b'_{i'}\}\in E(G)$, with  $v_0=p$, $v'_0=p$, $\alpha=(a_ib_i)_{i\in\{0,...,n-1\}}$, $\alpha'=(a'_{i'}b'_{i'})_{i\in\{0,...,n'-1\}}$ and $v_{n}=v_{n'}$.
We write $\hat{\alpha}$ for the equivalence class of $\alpha$ with respect to $\equiv_X$.
\end{definition}

It is often useful to undo the modulo, i.e. to obtain a canonical instance $(G(X),\varepsilon)$ of the equivalence class $X$.
\begin{definition}[Associated graph]\label{def:associatedgraph}
Let $X$ be a pointed graph modulo. Let $G(X)$ be the graph such that:
\begin{itemize}
\item[$\bullet$] The set of vertices $V(G(X))$ is the set of equivalence classes of $L(X)$;
\item[$\bullet$] The edge $\{\hat{\alpha}\port a,\hat{\beta}\port b\}$ is in $E(G(X))$ if and only if $\alpha.ab \in L$ and $\alpha.ab\equiv_X \beta$, for all $\alpha\in \hat{\alpha}$ and $\beta\in \hat{\beta}$.
\end{itemize}
We define the {\em associated graph} to be $G(X)$.
\end{definition}

\noindent {\em Conventions.} The following are three presentations of the same mathematical object:
\begin{itemize}
\item[$\bullet$] a graph modulo $X$,
\item[$\bullet$] its associated graph $G(X)$
\item[$\bullet$] the algebraic structure $\langle L(X),\equiv_X\rangle$
\end{itemize}
Each vertex of this mathematical object can thus be designated by
\begin{itemize}
\item[$\bullet$]  $\hat{\alpha}$ an equivalence class of $L(X)$, i.e. the set of all paths leading to this vertex starting from $\hat{\varepsilon}$,
\item[$\bullet$] or more directly by $\alpha$ an element of an equivalence class $\hat{\alpha}$ of $X$, i.e. a particular path leading to this vertex starting from $\varepsilon$.
\end{itemize}
These two remarks lead to the following mathematical conventions, which we adopt for convenience. In the paper:
\begin{itemize}
\item[$\bullet$] $\hat{\alpha}$ and $\alpha$ are no longer distinguished unless otherwise specified. The latter notation is given the meaning of the former. We speak of a ``vertex'' $\alpha$ in $V(X)$ (or simply $\alpha\in X$).
\item[$\bullet$] It follows that `$\equiv_X$' and `$=$' are no longer distinguished unless otherwise specified. The latter notation is given the meaning of the former. I.e. we speak of ``equality of vertices'' $\alpha=\beta$ (when strictly speaking we just have $\hat{\alpha}=\hat{\beta}$).
\end{itemize}

\subsection{Operations ove pointed Graphs modulo}\label{app:operationsmodulo}

\noindent {\em Subdisks}. For a pointed graph $(G,p)$ non-modulo:
\begin{itemize}
\item[$\bullet$] the neighbours of radius $r$ are just those vertices which can be reached in $r$ steps starting from the pointer $p$;
\item[$\bullet$] the disk of radius $r$, written $G^r_p$, is the subgraph induced by the neighbours of radius $r+1$, with labellings restricted to the neighbours of radius $r$ and the edges between them, and pointed at $p$.
\end{itemize}
For a graph modulo, on the other hand, the analogous operation is:
\begin{definition}[Disk]
Let $X\in {\cal X}_{\Sigma,\Delta,\ports}$ be a pointed graph modulo and $G$ its associated graph. 
Let $X^r$ be $\widetilde{G^r_\varepsilon}$.
The graph modulo $X^r\in {\cal X}_{\Sigma,\Delta,\ports}$ is referred to as the {\em disk of radius $r$} of $X$. The {\em set of disks of radius $r$} with states $\Sigma,\Delta$ and ports $\ports$ is written ${\cal X}^r_{\Sigma,\Delta,\ports}$.
\end{definition}

\begin{definition}[Size]\label{def:size}
Let $X\in {\cal X}_{\Sigma,\Delta,\ports}$ be a pointed graph modulo. We say that a vertex $u\in X$ has size less or equal to $r+1$, and write $|u|\leq r+1$, if and only if $u\in X^r$.  
\end{definition}

\noindent {\em Shifts} are a notation for the graph where vertices are named relatively to some other pointer vertex $u$.
\begin{definition}[Shift]\label{def:shift}
Let $X\in {\cal X}_{\Sigma,\Delta,\ports}$ be a pointed graph modulo and $G$ its associated graph. 
Consider $u\in X$ or $X^r$ for some $r$, and consider the pointed graph $(G,u)$, which is the same as $(G,\varepsilon)$ but with a different pointer. Let $X_u$ be $\tili{(G,u)}$. The pointed graph modulo $X_u$ is referred to as {\em $X$ shifted by $u$}.\\ 
\end{definition}

\subsection{Operations over pointed Graphs non-modulo}\label{app:operationsnonmodulo}

\begin{definition}[Shift isomorphism]\label{def:shiftiso}
Let $X\in {\cal X}_{\ports}$ be a pointed graph modulo.
Let $G\in {\cal G}_{\ports}$ be a graph that has vertices that are disjoint subsets of $V(X)$ or $V(X^r)$ for some $r$. 
Consider $u \in X$. 
Let $R$ be the isomorphism from $V(X)$ to $V(X_u)$ mapping $v\mapsto \overline{u}.v$, for any $v\in V(X)$ or $V(X^r)$. Extend this bijection pointwise to act over subsets of $V(X)$, and let $\overline{u}.G$ to be $RG$.
The graph $\overline{u}.G$ has vertices that are disjoint subsets of $V(X_u)$, it is referred to as {\em $G$ shifted by $u$}.
The definition extends to labelled graphs.
\end{definition}

\noindent We need the standard \cite{BFHAmalgamation,LoweAlgebraic} notion of {\em union} of graphs, and for this purpose we need a notion of {\em consistency} between the operands of the union: 
\begin{definition}[Consistency]\label{def:consistency}
Let $X\in {\cal X}_{\ports}$ be a pointed graph modulo. Let $G$ be a labelled graph $(G,\sigma, \delta)$, and $G'$ be a labelled graph $(G',\sigma', \delta')$, each one having vertices that are pairwise disjoint subsets of $V(X)$.  The graphs are said to be {\em consistent} if and only if:
\begin{itemize}
\item[(i)] $\forall x\in G\,\forall x'\in G'\quad x\cap x'\neq\emptyset \Rightarrow x=x'$,
\item[(ii)] $\forall x,y\in G\,\forall x',y'\in G'\,\forall a,a',b,b' \in \ports$\\ 
$\displaystyle{(\{x\port a,y\port b\} \in E(G) \wedge \{x'\port a',y'\port b'\} \in E(G') \wedge x=x' \wedge a=a')}$\\${\quad\Rightarrow (b=b' \wedge y=y')}$,
\item[(iii)] $\forall x,y\in G\,\forall x',y'\in G'\,\forall a,b \in \ports \quad x=x' \Rightarrow \delta(\{x\port a,y\port b\})=\delta'(\{x'\port a,y'\port b\})$ when both are defined,
\item[(iv)] $\forall x\in G\,\forall x'\in G'\quad  x=x' \Rightarrow \sigma(x)=\sigma'(x')$ when both are defined.
\end{itemize}
They are said to be {\em trivially consistent} if and only if for all $x\in G$, $x'\in G'$ we have $x\cap x'=\emptyset$.
\end{definition}
(The consistency conditions aim at making sure that both graphs ``do not disagree''. Indeed: $(iv)$ means that ``if $G$ says that vertex $x$ has label $\sigma(x)$, $G'$ should either agree or have no label for $x$''; $(iii)$ means that ``if $G$ says that edge $e$ has label $\delta(e)$, $G'$ should either agree or have no label for $e$''; $(ii)$ means that ``if $G$ says that starting from vertex $x$ and following port $a$ leads to $y$ via port $b$, $G'$ should either agree or have no edge on port $x\port a$''. Condition $(i)$ is in the same spirit: it requires that $G$ and $G'$, if they have a vertex in common, then they must fully agree on its name. Remember that vertices of $G$ and $G'$ are disjoint subsets of $V(X).S$. If one wishes to take the union of $G$ and $G'$, one has to enforce that the vertex names are still disjoint subsets of $V(X).S$. Trivial consistency arises when $G$ and $G'$ have no vertex in common: thus, they cannot disagree on any of the above.)

\begin{definition}[Union]\label{def:union}
Let $X\in {\cal X}_{\ports}$ be a pointed graph modulo. Let $G$ be a labelled graph $(G,\sigma, \delta)$, and $G'$ be a labelled graph $(G',\sigma', \delta')$, each one having vertices that are pairwise disjoint subsets of $V(X).S$.
Whenever they are consistent, their {\em union} is defined. The resulting graph $G\cup G'$ is the labelled graph with vertices $V(G)\cup V(G')$, edges $E(G)\cup E(G')$, labels that are the union of the labels of $G$ and $G'$.
\end{definition}

\end{document}